\documentclass[a4paper, 11pt, fleqn]{article}

\usepackage{amsmath, amssymb, amsthm}
\usepackage{mathtools, thm-restate}
\usepackage{mathrsfs}
\usepackage[margin=1in]{geometry}
\usepackage{xcolor}
\usepackage{url}
\usepackage{comment}
\usepackage{tikz}
\usepackage{pgffor}
\usepackage{enumitem}
\usepackage{array}
\usepackage{longtable}
\usepackage{diagbox}
\usepackage{anyfontsize}
\usepackage{xfrac}
\usepackage[justification=centering]{caption}
\usepackage{subcaption}
\usepackage[hypertexnames=false,final]{hyperref}
\usepackage{algorithm, algpseudocode}
\usepackage[capitalize,sort]{cleveref}
\usetikzlibrary{decorations.pathreplacing,arrows.meta,calc}

\usepackage{amsmath}

\let\epsilon\varepsilon

\newcommand*{\thmdep}[2]{}
\newcommand*{\thmdepcref}[2]{\cref{#1}}

\newcommand*{\Th}{^{\textrm{th}}}
\newcommand*{\WLoG}{Without loss of generality}

\newcommand*{\wLoG}{without loss of generality}

\let\eps\epsilon
\newcommand*{\itild}{\tilde{\imath}}
\newcommand*{\ihat}{\hat{\imath}}

\newcommand*{\Ihat}{\widehat{I}}
\newcommand*{\Ical}{\mathcal{I}}

\newcommand*{\Jhat}{\widehat{J}}

\newcommand*{\lavec}[1]{\mathbf{#1}}  %

\newcommand*{\defeq}{:=}

\newcommand*{\floor}[1]{\left\lfloor #1 \right\rfloor}
\newcommand*{\ceil}[1]{\left\lceil #1 \right\rceil}

\newcommand*{\smallceil}[1]{\lceil #1 \rceil}

\DeclareMathOperator*{\E}{E}

\DeclareMathOperator*{\argmin}{argmin}
\DeclareMathOperator*{\argmax}{argmax}

\DeclareMathOperator{\opt}{opt}

\DeclareMathOperator{\Sum}{sum}
\DeclareMathOperator{\vol}{vol}

\DeclareMathOperator{\size}{size}

\DeclareMathOperator{\LP}{LP}

\usetikzlibrary{3d}
\newcommand*{\config}{configuration}

\newcommand*{\asymAppx}{asymptotic-approximate}
\newcommand*{\appx}{approximate}
\newcommand*{\dff}{weighting function}

\newcommand*{\DFF}{Weighting Function}

\newcommand*{\Ccal}{\mathcal{C}}
\newcommand*{\Chat}{\widehat{C}}
\newcommand*{\safed}{\texorpdfstring{$d$}{d}}

\newcommand*{\Htild}{\widetilde{H}}

\newcommand*{\Khat}{\widehat{K}}

\newcommand*{\Pbar}{\overline{P}}
\newcommand*{\xhat}{\widehat{x}}
\newcommand*{\yhat}{\widehat{y}}

\newcommand*{\Icalhat}{\widehat{\mathcal{I}}}

\newcommand*{\Null}{\texttt{null}}
\newcommand*{\lo}{^{\mathrm{(lo)}}}
\newcommand*{\hi}{^{\mathrm{(hi)}}}
\newcommand*{\best}{\mathrm{best}}
\newcommand*{\ceildeltsq}{\smallceil{1/\delta^2}}

\newcommand*{\optdsp}{\operatorname{opt}_{d\mathrm{SP}}}
\newcommand*{\optdbp}{\operatorname{opt}_{d\mathrm{BP}}}

\newcommand*{\optdmcbp}{\operatorname{opt}_{d\mathrm{MCBP}}}
\newcommand*{\optdmcsp}{\operatorname{opt}_{d\mathrm{MCSP}}}
\newcommand*{\optdmcks}[1][d]{\operatorname{opt}_{{#1}\mathrm{MCKS}}}

\newcommand*{\assortSet}{\Psi}
\DeclareMathOperator{\type}{type}
\DeclareMathOperator{\btype}{btype}
\DeclareMathOperator{\last}{last}

\DeclareMathOperator{\sopt}{sopt}
\DeclareMathOperator{\DLP}{DLP}

\newcommand*{\wfk}[1][k]{\widetilde{f}_{#1}}
\newcommand*{\wHk}[1][k]{\widetilde{H}_{#1}}
\DeclareMathOperator{\flatten}{flat}
\DeclareMathOperator{\DLPbound}{DLPbound}

\DeclareMathOperator{\proj}{proj}
\DeclareMathOperator{\level}{level}
\DeclareMathOperator{\canShelv}{can-shelv}
\DeclareMathOperator{\canShelvHyp}{\hyperref[defn:hgap:can-shelv]{\canShelv}}
\DeclareMathOperator{\reduce}{reduce}
\DeclareMathOperator{\smallArea}{smallArea}

\newcommand*{\hdhk}[1][k]{\operatorname{\mathtt{HDH}}_{#1}}
\newcommand*{\hdhksp}[1][k]{\operatorname{\mathtt{HDH-SP}}_{#1}}
\newcommand*{\hdhkspHyp}[1][k]{\operatorname{\hyperref[algo:hdhksp]{\mathtt{HDH-SP}}}_{#1}}

\newcommand*{\hdhkunit}{\operatorname{\mathtt{HDH-unit-pack}}_k}
\newcommand*{\hdhkunitHyp}{\operatorname{\hyperref[sec:hdhk-prelims:hdhkunit]{\mathtt{HDH-unit-pack}}}_k}
\newcommand*{\hdhknf}[1][k]{\operatorname{\mathtt{HDH-NF}}_{#1}}
\newcommand*{\fhk}{\operatorname{\mathtt{fullh}}_k}
\newcommand*{\hdhks}{\operatorname{\mathtt{HDH-KS}}}

\DeclareMathOperator{\round}{\mathtt{round}}

\newcommand*{\hgapk}{\operatorname{\mathtt{HGaP}}_k}

\DeclareMathOperator{\chooseAndPack}{\mathtt{choose-and-pack}}
\DeclareMathOperator{\simpleChooseAndPack}{\mathtt{simple-choose-and-pack}}
\DeclareMathOperator{\inflate}{\mathtt{inflate}}
\DeclareMathOperator{\guessShelves}{\mathtt{guess-shelves}}

\usepackage{iftex}
\usepackage{url}
\usepackage{mathtools}
\usepackage{comment}

\hypersetup{
    colorlinks,
    linkcolor={red!50!black},
    citecolor={red!50!black},
    urlcolor={blue!50!black}
}

\makeatletter
\g@addto@macro{\UrlBreaks}{%
\do\/%
\do\a\do\b\do\c\do\d\do\e\do\f\do\g\do\h\do\i\do\j\do\k\do\l\do\m%
\do\n\do\o\do\p\do\q\do\r\do\s\do\t\do\u\do\v\do\w\do\x\do\y\do\z%
\do\A\do\B\do\C\do\D\do\E\do\F\do\G\do\H\do\I\do\J\do\K\do\L\do\M%
\do\N\do\O\do\P\do\Q\do\R\do\S\do\T\do\U\do\V\do\W\do\X\do\Y\do\Z%
\do\0\do\1\do\2\do\3\do\4\do\5\do\6\do\7\do\8\do\9%
}
\makeatother

\renewcommand*{\defeq}{\coloneqq}
\newcolumntype{L}{>{$\displaystyle}l<{$}}
\algnewcommand{\LineComment}[1]{\State \textcolor{gray}{\texttt{//} \textit{#1}}}

\newcommand*{\acknowledgements}[1]{\paragraph{Acknowledgements.} #1}

\newtheorem{theorem}{Theorem}
\newtheorem{definition}{Definition}
\newtheorem{property}[definition]{Property}
\newtheorem{corollary}{Corollary}[theorem]
\newtheorem{lemma}[theorem]{Lemma}

\newtheorem{observation}[theorem]{Observation}
\newtheorem{transformation}[definition]{Transformation}

\crefname{claim}{Claim}{Claims}
\crefname{property}{Property}{Properties}
\crefname{observation}{Observation}{Observations}
\crefname{transformation}{Transformation}{Transformations}

\title{Harmonic Algorithms for Packing\\$d$-dimensional Cuboids Into Bins}
\author{Eklavya Sharma\\
Department of Computer Science and Automation\\
Indian Institute of Science, Bengaluru.\\
\texttt{eklavyas@iisc.ac.in}}
\date{\empty}

\begin{document}

\maketitle

\begin{abstract}
We explore approximation algorithms for the $d$-dimensional
geometric bin packing problem ($d$BP).
Caprara \texorpdfstring{\cite{caprara2008}}{(MOR 2008)}
gave a \emph{harmonic-based} algorithm for $d$BP
having an asymptotic approximation ratio (AAR) of
\texorpdfstring{$T_{\infty}^{d-1}$ (where $T_{\infty} \approx 1.691$)}{1.692^(d-1)}.
However, their algorithm doesn't allow items to be rotated.
This is in contrast to some common applications of $d$BP,
like packing boxes into shipping containers.
We give approximation algorithms for $d$BP when items can be
orthogonally rotated about all or a subset of axes.
We first give a fast and simple harmonic-based algorithm having AAR
\texorpdfstring{$T_{\infty}^{d}$}{1.692^d}.
We next give a more sophisticated harmonic-based algorithm,
which we call \texorpdfstring{$\hgapk$}{HGaP}, having AAR
\texorpdfstring{$T_{\infty}^{d-1}(1+\eps)$}{(1+epsilon)*1.692^(d-1)}.
This gives an AAR of roughly $2.860 + \eps$ for 3BP with rotations,
which improves upon the best-known AAR of $4.5$.
In addition, we study the \emph{multiple-choice} bin packing problem
that generalizes the rotational case.
Here we are given $n$ sets of $d$-dimensional cuboidal items and we have to
choose exactly one item from each set and then pack the chosen items.
Our algorithms also work for the multiple-choice bin packing problem.
We also give fast and simple approximation algorithms for the multiple-choice
versions of $d$D strip packing and $d$D geometric knapsack.

\end{abstract}

\acknowledgements{I want to thank my advisor, Prof.~Arindam Khan,
for his valuable comments, and Arka Ray for helpful suggestions.}

\section{Introduction}

Packing of rectangular and cuboidal items is a fundamental problem in computer science,
mathematics, and operations research. Packing problems find numerous applications in practice,
e.g., physical packing of concrete 3D items during
storage or transportation~\cite{bortfeldt2013constraints},
cutting prescribed 2D pieces from cloth or metal sheet
while minimizing the waste~\cite{gilmore1961linear}, etc.
In this paper, we study packing of $d$-dimensional ($d$D) cuboidal items (for $d \ge 2$).

Let $I$ be a set of $n$ number of $d$D cuboidal items,
where each item has length at most one in each dimension.
A feasible packing of items into a $d$D cuboid is a packing where
items are placed inside the cuboid parallel to the axes without any overlapping.
A $d$D {\em unit cube} is a $d$D cuboid of length one in each dimension.
In the $d$D bin packing problem ($d$BP), we have to compute
a feasible packing of $I$ (without rotating the items)
into the minimum number of bins that are $d$D unit cubes.
Let $\optdbp(I)$ denote the minimum number of bins needed to pack $I$.

$d$BP is NP-hard, as it generalizes the classic bin packing problem~\cite{coffman2013bin}.
Thus, we study approximation algorithms.
For $d$BP, the worst-case approximation ratio usually occurs only for
{\em small} pathological instances.
Thus, the standard performance measure is the asymptotic approximation ratio (AAR).
For an algorithm $\mathcal{A}$, AAR is defined as:
\[ \lim\limits_{m \to \infty} \quad \sup_{I \in \mathtt{I}:\, \opt(I) = m}
    \quad \frac{\mathcal{A}(I)}{\opt(I)}, \]
where $\mathtt{I}$ is the set of all problem instances.
$\mathcal{A}(I)$ and $\opt(I)$ are the number of bins used by
$\mathcal{A}$ and the optimal algorithm, respectively, on $I$.

Coffman et al.~\cite{coffman1980performance} initiated the study of approximation
algorithms for rectangle packing. They studied packing algorithms such as
First-Fit Decreasing Height (FFDH) and Next-Fit Decreasing Height (NFDH).
In his seminal paper, Caprara~\cite{caprara2008} devised a polynomial-time algorithm
for $d$BP called $\hdhk$ (Harmonic Decreasing Height),
where $k \in \mathbb{Z}$ is a parameter to the algorithm.
$\hdhk$ has AAR equal to $T_k^{d-1}$,
where $T_k$ is a decreasing function of $k$ and
$T_{\infty} \defeq \lim_{k \to \infty} T_k \approx 1.691$.
The algorithm $\hdhk$ is based on an extension of the harmonic algorithm~\cite{leelee} for 1BP.

A limitation of $\hdhk$ is that it does not allow rotation of items.
This is in contrast to some real-world problems,
like packing boxes into shipping containers ($d=3$),
where items can often be rotated orthogonally,
i.e., $90^{\circ}$ rotation around all or a subset of axes~\cite{baldi2012three, stoyan2014packing}.
Orientation constraints may sometimes limit the vertical orientation of a box
to one dimension (``This side up'')
or to two (of three) dimensions (e.g., long but low and
narrow box should not be placed on its smallest surface).
These constraints are introduced to deter goods
and packaging from being damaged and to ensure the stability of the load.
One of our primary contributions is presenting variants of $\hdhk$
that work for generalizations of $d$BP
that capture the notion of orthogonal rotation of items.

\subsection{Prior Work}

For 2BP, Bansal et al.~\cite{rna} obtained AAR of $T_{\infty}+\eps$ even for the case with rotations,
using a more sophisticated algorithm that used properties of harmonic rounding.
Then there has been a series of improvements \cite{rna, JansenP2013} culminating with
the present best AAR of 1.406~\cite{bansal2014binpacking}, for both the
cases with and without orthogonal rotations.
Bansal et al.~\cite{bansal2004} showed that $d$BP is APX-hard, even for $d=2$.
They also gave an asymptotic PTAS for $d$BP where all items are $d$D squares.

Closely related to $d$BP is the $d$D strip packing problem ($d$SP), where
we have to compute a packing of $I$ (without rotating the items)
into a $d$D cuboid (called a strip) that has length one in the first $d-1$ dimensions
and has the minimum possible length (called height) in the $d\Th$ dimension.

For 2SP, an asymptotic PTAS was given by Kenyon and R\'emila~\cite{kenyon1996strip}.
Jansen and van Stee~\cite{jansen2005strip} extended this to the case with orthogonal rotations.
For 3SP, when rotations are not allowed, Bansal et al.~\cite{BansalHISZ13} gave a harmonic-based
algorithm achieving AAR of $T_{\infty}+\eps$.
Recently, this has been improved to $1.5+\eps$~\cite{JansenPS14}.
Miyazawa and Wakabayashi~\cite{miyazawa2009three} studied 3SP
and 3BP when rotations are allowed, and
gave algorithms with AAR 2.64 and 4.89, respectively.
Epstein and van Stee~\cite{EpsteinS06a} gave an improved AAR of 2.25 and 4.5
for 3SP and 3BP with rotations, respectively.
The $\hdhk$ algorithm also works for $d$SP and has an AAR of $T_k^{d-1}$.
For online $d$BP, there are harmonic-based
$T_\infty^d$-asymptotic-approximation algorithms \cite{EpsteinS05, CsirikV93},
which are optimal for $O(1)$ memory algorithms.

\subsection{Multiple-Choice Packing}

We will now define the $d$D multiple-choice bin packing problem ($d$MCBP).
This generalizes $d$BP and captures the notion of orthogonal rotation of items.
This perspective will be helpful in designing algorithms for the rotational case.
In $d$MCBP, we're given a set $\Ical = \{I_1, I_2, \ldots, I_n\}$,
where for each $j$, $I_j$ is a set of items, henceforth called an {\em itemset}.
We have to pick exactly one item from each itemset and pack those items
into the minimum number of bins. See \cref{fig:2mcbp} for an example of 2MCBP.

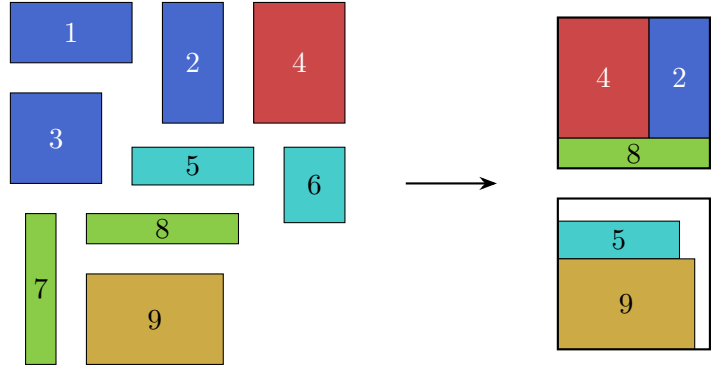
\begin{figure}[htb]
\centering
\begin{tikzpicture}[
item/.style={draw},
bin/.style={draw,thick},
myarrow/.style={->,>={Stealth},thick},
scale=0.8,
]
\definecolor{myblue}{HTML}{4768CC}
\definecolor{myred}{HTML}{CC4747}
\definecolor{mycyan}{HTML}{47CCCC}
\definecolor{mygreen}{HTML}{8ACC47}
\definecolor{myyellow}{HTML}{CCAB47}
\begin{scope}
\path[item,fill=myblue]
    (0, 0) rectangle +(2, -1) node[text=white,pos=0.5] {1}
    (2.5, 0) rectangle +(1, -2) node[text=white,pos=0.5] {2}
    (0, -1.5) rectangle +(1.5, -1.5) node[text=white,pos=0.5] {3};
\path[item,fill=myred]
    (4, 0) rectangle +(1.5, -2) node[text=white,pos=0.5] {4};
\path[item,fill=mycyan]
    (2, -2.4) rectangle +(2, -0.625) node[pos=0.5] {5}
    (4.5, -2.4) rectangle +(1, -1.25) node[pos=0.5] {6};
\path[item,fill=mygreen]
    (0.25, -3.5) rectangle +(0.5, -2.5) node[pos=0.5] {7}
    (1.25, -3.5) rectangle +(2.5, -0.5) node[pos=0.5] {8};
\path[item,fill=myyellow]
    (1.25, -4.5) rectangle +(2.25, -1.5) node[pos=0.5] {9};
\end{scope}
\draw[myarrow] (6.5, -3) -- (8,-3);
\begin{scope}[xshift=9cm,yshift=-0.25cm]
\path[item,fill=myred] (0, 0) rectangle +(1.5, -2) node[text=white,pos=0.5] {4};
\path[item,fill=myblue] (1.5, 0) rectangle +(1, -2) node[text=white,pos=0.5] {2};
\path[item,fill=mygreen] (0, -2) rectangle +(2.5, -0.5) node[pos=0.5] {8};
\path[item,fill=mycyan] (0, -4) rectangle +(2, 0.625) node[pos=0.5] {5};
\path[item,fill=myyellow] (0, -4) rectangle +(2.25, -1.5) node[pos=0.5] {9};
\path[bin] (0, 0) rectangle +(2.5, -2.5);
\path[bin] (0, -3) rectangle +(2.5, -2.5);
\end{scope}
\end{tikzpicture}

\caption[2MCBP example]{2MCBP example: packing the input
$\Ical = \{\{1, 2, 3\}, \{4\}, \{5, 6\}, \{7, 8\}, \{9\}\}$ into two bins.
Here items of the same color belong to the same itemset.}
\label{fig:2mcbp}
\end{figure}

We can model rotations using multiple-choice packing:
Given a set $I$ of items, for each item $i \in I$,
create an itemset $I_i$ that contains all allowed orientations of $i$.
Then the optimal solution to $\Ical \defeq \{I_i: i \in I\}$
will tell us how to rotate and pack items in $I$.

Some algorithms for 2D bin packing with rotations assume that
the bin is square \cite{rna,JansenP2013,bansal2014binpacking}.
This assumption holds without loss of generality when rotations are forbidden,
because we can scale the items.
But if rotations are allowed, this won't work because
items $i_1$ and $i_2$ that are rotations of each other
may stop being rotations of each other after they are scaled.
Multiple-choice packing algorithms can be used in this case.
For each item $i \in I$, we will create an itemset $I_i$ that
contains scaled orientations of $i$.

Multiple-choice packing problems have been studied before.
Lawler gave an FPTAS for the multiple-choice knapsack problem~\cite{lawler1979fast}.
Patt-Shamir and Rawitz gave an algorithm for multiple-choice vector bin packing having
AAR $O(\log d)$ and a PTAS for multiple-choice vector knapsack~\cite{patt2012vector}.
Similar notions have been studied in the scheduling of
malleable or moldable jobs~\cite{ZhangJ07, Jansen12}.

\subsection{Our Contributions}

After the introduction of the harmonic algorithm for online 1BP
by Lee and Lee~\cite{leelee}, many variants have found
widespread use in multidimensional packing problems (both offline and online)
\cite{caprara2008,rna,BansalHISZ13,BaloghBDEL18,EpsteinS05,CsirikV93,han2011new,ramanan1989line,seiden2002online}.
They are also simple, fast, and easy to implement.
For example, among algorithms for 3SP, 2BP and 3BP with practical running time,
harmonic-based algorithms provide the best AAR.

In our work, we extend harmonic-based algorithms to $d$MCBP.
$d$MCBP subsumes the rotational case for geometric bin packing,
and we believe $d$MCBP is an important natural generalization
of geometric bin packing that may be of independent interest.

In \cref{sec:hdhk-prelims}, we describe ideas from $\hdhk$ \cite{caprara2008}
that help us devise harmonic-based algorithms for $d$MCBP.
In \cref{sec:fhk}, we show an $O(Nd + nd\log n)$-time algorithm for $d$MCBP,
called $\fhk$, having an AAR of $T_k^d$, where $n$ is the number of itemsets
and $N$ is the total number of items across all the $n$ itemsets.
$\fhk$ is a fast and simple algorithm that works in two stages:
In the first stage, we select the \emph{smallest} item from each itemset
(we will precisely define \emph{smallest} in \cref{sec:fhk}).
In the second stage, we pack the selected items into bins
using a variant of the $\hdhk$ algorithm.

In \cref{sec:hgap}, we show an algorithm for $d$MCBP, called $\hgapk$, having an AAR of
$T_k^{d-1}(1+\eps)$ and having a running time of $N^{O(1/\eps^2)}n^{(1/\eps)^{O(1/\eps)}}+O(Nd + nd\log n)$.
For $d \ge 3$, this matches the present best AAR for the case
where rotations are forbidden.
Also, for large $k$, this gives an AAR of roughly $T_{\infty}^2 \approx 2.860$
for 3D bin packing when orthogonal rotations are allowed,
which is an improvement over the previous best AAR of $4.5$~\cite{EpsteinS06a},
an improvement after fourteen years.

As harmonic algorithms are ubiquitous in bin packing, we expect
our results will have applications in other related problems.

Our techniques can be extended to some other packing problems,
like strip packing and geometric knapsack.
In \cref{sec:hdhk-sp}, we define the $d$D multiple-choice strip packing problem ($d$MCSP)
and extend Caprara's $\hdhk$ algorithm \cite{caprara2008} to $d$MCSP.
The algorithm has AAR $T_k^{d-1}$ and runs in time $O(Nd + nd\log n)$,
where $n$ is the number of itemsets and $N$ is the total number of items across all itemsets.
In \cref{sec:hdhks}, we define the $d$D multiple-choice knapsack problem ($d$MCKS),
and for any $0 < \eps < 1$, we show an $O(Nd + N\log N + Nn/\eps + nd\log n)$-time algorithm
that is $(1-\eps)3^{-d}$-approximate.

\section{Preliminaries}

Let $[n] \defeq \{1, 2, \ldots, n\}$.
For a set $X$, define $\Sum(X) \defeq \sum_{x \in X} x$.
For an $n$-dimensional vector $\lavec{v}$,
define $\Sum(\lavec{v}) \defeq \sum_{i=1}^n \lavec{v}_i$.
For a set $X \subseteq I$ of items and any function $f: I \mapsto \mathbb{R}$,
$f(X)$ is defined to be $\sum_{i \in X} f(i)$, unless stated otherwise.

The length of a $d$D item $i$ in the $j\Th$ dimension is denoted by $\ell_j(i)$.
Define $\vol(i) \defeq \prod_{j=1}^d \ell_j(i)$.
For a $d$D cuboid $i$, call the first $d-1$ dimensions \emph{base dimensions}
and call the $d\Th$ dimension \emph{height}.
For a set $I$ of items, $|I|$ is the number of items in $I$.
Let $|P|$ denote the number of bins used by a packing $P$ of items into bins.

\subsection{Multiple-Choice Packing}

Let $\Ical$ be a set of itemsets.
Define $\flatten(\Ical)$ to be the union of all itemsets in $\Ical$.

Let $K$ be a set of items that contains exactly one item from each itemset in $\Ical$.
Formally, for each itemset $I \in \Ical$, $|K \cap I| = 1$.
Then $K$ is called an \emph{assortment} of $\Ical$.
Let $\assortSet(\Ical)$ denote the set of all assortments of $\Ical$.
In $d$MCBP, given an input instance $\Ical$,
we have to select an assortment $K \in \assortSet(\Ical)$ and output a bin packing of $K$,
such that the number of bins used is minimized. Therefore,
$\optdmcbp(\Ical) = \min_{K \in \assortSet(\Ical)} \optdbp(K)$.

\section{Important Ideas from the \texorpdfstring{$\hdhk$}{HDHk} Algorithm}
\label{sec:hdhk-prelims}

In this section, we will describe some important ideas behind the $\hdhk$ algorithm
for $d$BP by Caprara~\cite{caprara2008}.
These ideas are the building blocks for our algorithms for $d$MCBP.

\subsection{\DFF{}s}

Fekete and Schepers~\cite{fekete2004} present a useful approach for
obtaining lower bounds on the optimal solution to bin packing problems.
Their approach is based on \emph{\dff{}s}.

\begin{definition}
$g: [0, 1] \mapsto [0, 1]$ is a weighting function iff
for all $m \in \mathbb{Z}_{>0}$ and $x \in [0, 1]^m$,
\[ \sum_{i=1}^m x_i \le 1 \implies \sum_{i=1}^m g(x_i) \le 1 \]
(Weighting functions are also called \emph{dual feasible functions} (DFFs)).
\end{definition}

\begin{restatable}{theorem}{rthmDffPack}
\label{thm:dff-pack}
Let $I$ be a set of $d$D items that can be packed into a bin.
Let $g_1, g_2, \ldots, g_d$ be \dff{}s.
For $i \in I$, define $g(i)$ as the item whose length is $g_j(\ell_j(i))$ in the $j\Th$ dimension,
for each $j \in [d]$.
Then $\{g(i): i \in I\}$ can be packed into a $d$D bin (without rotating the items).
\end{restatable}
\Cref{thm:dff-pack} is proved in \cref{sec:dff-trn}.

\subsection{The Harmonic Function}
\label{sec:hdhk-prelims:harmonic}

To obtain a lower-bound on $\optdbp(I)$ using \cref{thm:dff-pack},
Caprara~\cite{caprara2008} defined a function $f_k$.
For an integer constant $k \ge 3$, $f_k: [0, 1] \mapsto [0, 1]$ is defined as
\[ f_k(x) \defeq \begin{cases}
\frac{1}{q} & x \in \left(\left.\frac{1}{q+1}, \frac{1}{q}\right]\right. \textrm{ for } q \in [k-1]
\\ \frac{k}{k-2}x & x \le \frac{1}{k} \end{cases}. \]
$f_k$ was originally defined and studied by Lee and Lee~\cite{leelee} for their
online algorithm for 1BP, except that they used $k/(k-1)$ instead of $k/(k-2)$.
Define $\type_k: [0, 1] \mapsto [k]$ as
\[ \type_k(x) \defeq \begin{cases}
q & x \in \left(\left.\frac{1}{q+1}, \frac{1}{q}\right]\right. \textrm{ for } q \in [k-1]
\\ k & x \le \frac{1}{k} \end{cases}. \]

Define $T_k$ to be the smallest positive constant such that
$H_k(x) \defeq f_k(x)/T_k$ is a \dff{}. We call $H_k$ the \emph{harmonic \dff}.
We can efficiently compute $T_k$ as a function of $k$ using ideas from \cite{leelee}.
\Cref{table:tk-values} lists the values of $T_k$ for the first few $k$.
It can also be proven that $T_k$ is a decreasing function of $k$
and $T_{\infty} \defeq \lim_{k \to \infty} T_k \approx 1.6910302$.

\begin{table}[!ht]
\centering
\caption{Values of $T_k$.}
\begin{tabular}{|c|c|c|c|c|c|c|}
\hline
$k$ & 3 & 4 & 5 & 6 & 7 & $\infty$ \\
\hline
$T_k$ & 3 & 2 & $11/6 = 1.8\overline{3}$ & $7/4 = 1.75$
& $26/15 = 1.7\overline{3}$ & $\approx 1.6910302$ \\
\hline
\end{tabular}
\label{table:tk-values}
\end{table}

For a $d$D cuboid $i$, define $f_k(i)$ to be the cuboid
whose length is $f_k(\ell_j(i))$ in the $j\Th$ dimension, for each $j \in [d]$.
For a set $I$ of $d$D cuboids, let $f_k(I) \defeq \{f_k(i): i \in I\}$.
Similarly define $H_k(i)$ and $H_k(I)$.
Define $\type(i)$ to be a $d$-dimensional vector whose $j\Th$ component is $\type_k(\ell_j(i))$.
Note that there can be at most $k^d$ different values of $\type(i)$.
Sometimes, for the sake of convenience, we may express $\type(i)$ as an integer in $[k^d]$.

\begin{theorem}
\label{thm:fvol-bp}
For a set of $I$ of $d$D items, $\vol(f_k(I)) \le T_k^d \optdbp(I)$.
\end{theorem}
\begin{proof}
Let $m \defeq \optdbp(I)$. Let $J_j$ be the items in the
$j\Th$ bin in the optimal bin packing of $I$.
By \cref{thm:dff-pack} and because $H_k$ is a \dff{},
$H_k(J_j)$ fits in a bin. Therefore,
\[ \vol(f_k(I))
= \sum_{j=1}^m T_k^d \vol(H_k(J_j))
\le \sum_{j=1}^m T_k^d
= T_k^d \optdbp(I).
\qedhere \]
\end{proof}

\subsection{The \texorpdfstring{$\hdhkunit$}{HDH-unit-pack} Subroutine}
\label{sec:hdhk-prelims:hdhkunit}

From the $\hdhk$ algorithm by Caprara~\cite{caprara2008},
we extracted out a useful subroutine, which we call $\hdhkunit$,
that satisfies the following useful property:
\begin{property}
\label{prop:hdhkunit}
The algorithm $\hdhkunit^{[t]}(I)$ takes a \emph{sequence} $I$ of $d$D items such that
all items have type $t$ and $\vol(f_k(I-\{\last(I)\})) < 1$
(here $\last(I)$ is the last item in sequence $I$).
It returns a packing of $I$ into a single $d$D bin
in $O(nd\log n)$ time, where $n \defeq |I|$.
\end{property}

The design of $\hdhkunit$ and its correctness can be inferred from
Lemma 4.1 in \cite{caprara2008}.
We use $\hdhkunit$ as a black-box subroutine in our algorithms,
i.e., we only rely on \cref{prop:hdhkunit};
we don't need to know anything else about $\hdhkunit$.
Nevertheless, for the sake of completeness, in \cref{sec:hdhkunit},
we give a complete description of $\hdhkunit$ and prove its correctness.

\section{Fast and Simple Algorithm for \safed{}MCBP
(\texorpdfstring{$\fhk$}{fullh\_k})}
\label{sec:fhk}

We will now describe an algorithm for $d$BP called the \emph{full-harmonic algorithm}
($\fhk$). We will then extend it to $d$MCBP.
The $\fhk$ algorithm works by first partitioning the items based on their $\type$ vector
(type vector is defined in \cref{sec:hdhk-prelims:harmonic}).
Then for each partition, it repeatedly picks the smallest prefix $J$
such that $\vol(f_k(J)) \ge 1$ and packs $J$ into a $d$D bin using $\hdhkunit$.
See \cref{algo:fhk} for a more precise description of $\fhk$.
Note that $\fhk(I)$ has a running time of $O(|I|d\log |I|)$.

\begin{algorithm}[!ht]
\caption{$\fhk(I)$: Returns a bin packing of $d$D items $I$.}
\label{algo:fhk}
\begin{algorithmic}[1]
\State Let $P$ be an empty list.
\For{each $\type$ $t$}
    \State $I^{[t]} = \{i \in I: \type(i) = t\}$.
    \While{$|I^{[t]}| > 0$}
        \State Find $J$, the smallest prefix of $I^{[t]}$ such that
            $J = I^{[t]}$ or $\vol(f_k(J))) \ge 1$.
        \State $B = \hdhkunitHyp^{[t]}(J)$.
            \Comment{$B$ is a packing of $J$ into a $d$D bin}.
        \State Append $B$ to the list $P$.
        \State Remove $J$ from $I^{[t]}$.
    \EndWhile
\EndFor
\State \Return the list $P$ of bins.
\end{algorithmic}
\end{algorithm}

\begin{theorem}
\label{thm:fhk-fvol}
The number of bins used by $\fhk(I)$ is less than $Q + \vol(f_k(I))$,
where $Q$ is the number of distinct $\type$s of items (so $Q \le k^d$).

\end{theorem}
\begin{proof}
Let $I^{[t]}$ be the items in $I$ of type $t$.
Suppose $\fhk(I)$ uses $m^{[t]}$ bins to pack $I^{[t]}$.
For each type $t$, the first $m^{[t]}-1$ bins have $\vol\cdot f_k$ at least 1,
so $\vol(f_k(I^{[t]})) > m^{[t]} - 1$.
Therefore, total number of bins used is
$\sum_{t=1}^Q m^{[t]} < \sum_{t=1}^Q (1 + \vol(f_k(I^{[t]})))
= Q + \vol(f_k(I))$.
\end{proof}

\begin{lemma}[Corollary to \cref{thm:fhk-fvol,thm:fvol-bp}]
\label{thm:fhk-appx}
$\fhk(I)$ uses less than $Q + T_k^d\optdbp(I)$ bins,
where $Q$ is the number of distinct item $\type$s.
\end{lemma}

\begin{theorem}
Let $\Ical$ be a $d$MCBP instance.
Let $\Khat \defeq \{\argmin_{i \in I} \vol(f_k(i)): I \in \Ical\}$,
i.e., $\Khat$ is the assortment obtained by picking from each itemset
the item $i$ having the minimum value of $\vol(f_k(i))$.
Then the number of bins used by $\fhk(\Khat)$ is less than $Q + T_k^d\optdmcbp(\Ical)$,
where $Q$ is the number of distinct $\type$s of items in $\flatten(\Ical)$ (so $Q \le k^d$).
\end{theorem}
\begin{proof}
For any assortment $K$, $\vol(f_k(\Khat)) \le \vol(f_k(K))$.
Let $K^*$ be the assortment in an optimal packing of $\Ical$.
By \cref{thm:fhk-fvol,thm:fvol-bp}, the number of bins used by
$\fhk(\Khat)$ is less than
\[ Q + \vol(f_k(\Khat))
\le Q + \vol(f_k(K^*))
\le Q + T_k^d\optdbp(K^*).
= Q + T_k^d\optdmcbp(\Ical) \qedhere \]
\end{proof}

Let $N \defeq |\flatten(\Ical)|$ and $n \defeq |\Ical|$.
We can find $\Khat$ in $O(Nd)$ time and compute $\fhk(\Khat)$ in $O(nd\log n)$ time.
This gives us an $O(Nd + nd\log n)$-time algorithm for $d$MCBP having AAR $T_k^d$.

\subsection{\safed{}BP with Rotations}
\label{sec:fhk-rot}

As mentioned before, we can solve the rotational version of $d$BP by reducing it to $d$MCBP.
Specifically, for each item $i$ in the $d$BP instance,
we create an itemset containing all orientations of $i$,
and we pack the resulting $d$MCBP instance using $\fhk$.
Since an item can have up to $d!$ allowed orientations,
this can take up to $O(nd! + nd\log n)$ time.
Hence, the running time is large when $d$ is large.
However, we can do better for some special cases.

When the bin has the same length in each dimension, then for any item $i$,
$\vol(f_k(i))$ is independent of how we orient $i$.
Hence, we can orient the items $I$ arbitrarily and then pack them using $\fhk$
in $O(nd\log n)$ time.

Suppose there are no orientation constraints, i.e., all $d!$ orientations of each item are allowed.
Let $L_j$ be the length of the bin in the $j\Th$ dimension, for each $j \in [d]$.
To use $\fhk$ to pack $I$, we need to find the best orientation for each item $i \in I$,
i.e., we need to find a permutation $\pi$ for each item $i$ such that
$\prod_{j=1}^d f_k\left( \ell_{\pi_j}(i) / L_j\right)$ is minimized.
This can be formulated as a maximum-weight bipartite matching problem on a graph
with $d$ vertices in each partition: for every $u \in [d]$ and $v \in [d]$,
the edge $(u, v)$ has a non-negative weight of $-\log(f_k(\ell_u(i) / L_v))$.
So, using the Kuhn-Munkres algorithm \cite{munkres1957algorithms},
we can find the best orientation for each item in $O(d^3)$ time.
Hence, we can pack $I$ using $\fhk$ in $O(nd^3 + nd\log n)$ time.

\section{Better Algorithm for \safed{}MCBP (\texorpdfstring{$\hgapk$}{HGaP})}
\label{sec:hgap}

Here we will describe a $T_k^{d-1}(1+\eps)$-\asymAppx{} algorithm for $d$MCBP that is based on
$\hdhk$ and Lueker and Fernandez de la Vega's APTAS for 1BP~\cite{bp-aptas}.
We call our algorithm \emph{Harmonic Guess-and-Pack} ($\hgapk$).
This improves upon $\fhk$ that has AAR $T_k^d$.

\begin{definition}
For a $d$D item $i$, let $h(i) \defeq \ell_d(i)$,
$w(i) \defeq \prod_{j=1}^{d-1} f_k(\ell_j(i))$ and $a(i) \defeq w(i)h(i)$.
Let $\round(i)$ be a rectangle of height $h(i)$ and width $w(i)$.
For a set $X$ of $d$D items, define $w(X) \defeq \sum_{i \in X} w(i)$
and $\round(X) \defeq \{\round(i): i \in X\}$.
\end{definition}

For any $\eps > 0$, the algorithm $\hgapk(\Ical, \eps)$ returns
a bin packing of $\Ical$, where $\Ical$ is a set of $d$D itemsets.
$\hgapk$ first converts $\Ical$ to a set $\Icalhat$ of 2D itemsets.
It then computes $P_{\best}$, which is a \emph{structured} bin packing of $\Icalhat$
(we formally define \emph{structured} later).
Finally, it uses the algorithm $\inflate$ to convert $P_{\best}$
into a bin packing of the $d$D itemsets $\Ical$, where
$|\inflate(P_{\best})|$ is very close to $|P_{\best}|$.
See \cref{algo:hgap} for a more precise description.
We show that $|P_{\best}| \lessapprox T_k^{d-1}(1+\eps)\opt(\Ical)$,
which proves that $\hgapk$ has an AAR of $T_k^{d-1}(1+\eps)$.
This approach of converting items to 2D, packing them,
and then converting back to $d$D is very useful,
because most of our analysis is about how to compute a structured 2D packing, and
a packing of 2D items is easier to visualize and reason about than a packing of $d$D items.

\begin{algorithm}[!ht]
\caption{$\hgapk(\Ical, \eps)$: Returns a bin packing of $d$D itemsets $\Ical$,
where $\eps \in (0, 1)$.}
\label{algo:hgap}
\begin{algorithmic}[1]
\State Let $\delta \defeq \eps/(2+\eps)$.
\State \label{alg-line:hgap:round}$\Icalhat = \{\round(I): I \in \Ical\}$
\State Initialize $P_{\best}$ to \Null.
\For{$P \in \guessShelves(\Icalhat, \delta)$}
    \State $\Pbar = \chooseAndPack(\Icalhat, P, \delta)$
    \If{$\Pbar$ is not \Null{} and ($P_{\best}$ is \Null{} or $|\Pbar| \le |P_{\best}|$)}
        \State $P_{\best} = \Pbar$
    \EndIf
\EndFor
\State \Return $\inflate(P_{\best})$
\end{algorithmic}
\end{algorithm}

A 2D bin packing is said to be \emph{shelf-based} if items are packed into \emph{shelves}
and the shelves are packed into bins, where a shelf is a rectangle of width 1.
See \cref{fig:shelf-based-2} for an example.
A structured bin packing is a shelf-based bin packing
where the heights of the shelves satisfy some additional properties
(we describe these properties later).
The algorithm $\guessShelves$ repeatedly guesses the number and heights of shelves
and computes a structured packing $P$ of those shelves into bins.
Then for each packing $P$, the algorithm $\chooseAndPack(\Icalhat, P, \delta)$ \emph{tries to}
pack an assortment of $\Icalhat$ into the shelves in $P$ plus maybe one additional shelf.
If $\chooseAndPack$ succeeds, call the resulting bin packing $\Pbar$.
Otherwise, $\chooseAndPack$ returns \Null.
$P_{\best}$ is the value of $\Pbar$ with the minimum number of bins
across all guesses by $\guessShelves$.

\begin{figure}[!ht]
\centering
\begin{tikzpicture}[scale=0.8,
shelf-line/.style = {draw={black!30}},
item/.style = {fill={black!10}}
]
\draw[shelf-line] (0,1.5) -- (4,1.5);
\draw[shelf-line] (0,2.7) -- (4,2.7);
\draw[shelf-line] (0,3.5) -- (4,3.5);
\draw (0,0) rectangle (4,4);
\draw[item]
    (0.0,0) rectangle +(1.3,1.5)
    (1.3,0) rectangle +(1.6,1.3)
    (2.9,0) rectangle +(1.1,1.2);
\draw[item]
    (0.0,1.5) rectangle +(0.4,1.2)
    (0.4,1.5) rectangle +(2,1)
    (2.4,1.5) rectangle +(1.6,0.8);
\draw[item]
    (0,2.7) rectangle +(2,0.8)
    (2,2.7) rectangle +(1,0.6);
\end{tikzpicture}

\caption{An example of shelf-based packing with 3 shelves.}
\label{fig:shelf-based-2}
\end{figure}
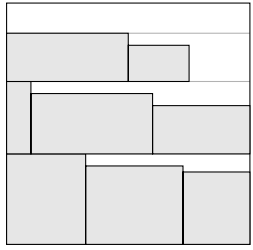

We prove that the AAR of $\hgapk$ is $T_k^{d-1}(1+\eps)$ by showing that
for some $P^* \in \guessShelves(\Icalhat, \delta)$,
we have $|P^*| \lessapprox T_k^{d-1}(1+\eps)\opt(\Ical)$
and $\chooseAndPack(\Icalhat, P^*, \delta)$ is not \Null.

We will now precisely define \emph{structured} bin packing
and state the main theorems on $\hgapk$.

\subsection{Structured Packing}

\begin{definition}[Slicing]
Slicing a 1D item $i$ is the operation of replacing it by items $i_1$ and $i_2$
such that $\size(i_1) + \size(i_2) = \size(i)$.

Slicing a rectangle $i$ using a vertical cut is the operation of replacing $i$
by two rectangles $i_1$ and $i_2$ where
$h(i) = h(i_1) = h(i_2)$ and $w(i) = w(i_1) + w(i_2)$.
Slicing $i$ using a horizontal cut is the operation of replacing $i$
by two rectangles $i_1$ and $i_2$ where
$w(i) = w(i_1) = w(i_2)$ and $h(i) = h(i_1) + h(i_2)$.
\end{definition}

\begin{definition}[Shelf-based $\delta$-fractional packing]
\label{defn:hgap:shelf-based-packing}
Let $\delta \in (0, 1)$ be a constant. Let $K$ be a set of rectangular items.
Items in $K_L \defeq \{i \in K: h(i) > \delta\}$ are said to be `$\delta$-large'
and items in $K_S \defeq K - K_L$ are said to be `$\delta$-small'.
A $\delta$-fractional bin packing of $K$ is defined to be a packing of $K$ into bins where
items in $K_L$ can be sliced (recursively) using vertical cuts only,
and items in $K_S$ can be sliced (recursively) using both horizontal and vertical cuts.

A \emph{shelf} is a rectangle of width 1 into which we can pack items
such that the bottom edge of each item in the shelf touches
the bottom edge of the shelf. A shelf can itself be packed into a bin.
A $\delta$-fractional bin packing of $K$ is said to be \emph{shelf-based} iff
(all slices of) all items in $K_L$ are packed into shelves,
the shelves are packed into the bins,
and items in $K_S$ are packed outside the shelves (and inside the bins).
Packing of items into a shelf $S$ is said to be \emph{tight} iff
the top edge of some item (or slice) in $S$ touches the top edge of $S$.
\end{definition}

\begin{definition}[Structured packing]
Let $K$ be a set of rectangles and
let $P$ be a packing of empty shelves into bins.
Let $H$ be the set of heights of shelves in $P$ (note that $H$ is not a multiset,
i.e., we only consider distinct heights of shelves).
Then $P$ is said to be \emph{structured} for $(K, \delta)$ iff
$|H| \le \ceildeltsq$ and each element in $H$ is
the height of some $\delta$-large item in $K$.

A shelf-based $\delta$-fractional packing of $K$ is said to be
\emph{structured} iff the shelves in the packing are structured for $(K, \delta)$.
Define $\sopt_{\delta}(K)$ to be the number of bins in the optimal structured
$\delta$-fractional packing of $K$.
\end{definition}

$\hgapk$ relies on the following key structural theorem.
We formally prove it in \cref{sec:hgap:struct} and give an outline of the proof here.
\begin{restatable}[Structural theorem]{theorem}{rthmHgapStruct}
\label{thm:hgap:struct}
Let $I$ be a set of $d$D items. Let $\delta \in (0, 1)$ be a constant. Then
$\sopt_{\delta}(\round(I)) < T_k^{d-1}(1+\delta)\optdbp(I) + \ceildeltsq + 1 + \delta$.
\end{restatable}
\begin{proof}[Proof outline]
Let $\Ihat \defeq \round(I)$. Let $\Ihat_L$ and $\Ihat_S$ be the
$\delta$-large and $\delta$-small items in $\Ihat$, respectively.

We give a simple greedy algorithm to pack $\Ihat_L$ into shelves.
Let $J$ be the shelves output by this algorithm.
We can treat $J$ as a 1BP instance,
and $\Ihat_S$ as a sliceable 1D item of size $a(\Ihat_S)$.
We prove that an optimal 1D bin packing of $J \cup \Ihat_S$ gives us
an optimal shelf-based $\delta$-fractional packing of $\Ihat$.

We use linear grouping by Lueker and Fernandez de la Vega~\cite{bp-aptas}.
We partition $J$ into linear groups of size $\floor{\delta\size(J)} + 1$ each.
Let $h_j$ be the height of the first 1D item in the $j\Th$ group.
Let $J\hi$ be the 1BP instance obtained by rounding up the height of each item
in the $j\Th$ group to $h_j$ for all $j$.
Then $J\hi$ contains at most $\ceildeltsq$ distinct sizes,
so the optimal packing of $J\hi \cup \Ihat_S$ gives us a
structured $\delta$-fractional packing of $\Ihat$.
Therefore, $\sopt_{\delta}(\Ihat) \le \opt(J\hi \cup \Ihat_S)$.
Let $J\lo$ be the 1BP instance obtained by rounding down the height of each item
in the $j\Th$ group to $h_{j+1}$ for all $j$.
We prove that $J\lo$ contains at most $\ceildeltsq-1$ distinct sizes and that
$\opt(J\hi \cup \Ihat_S) < \opt(J\lo \cup \Ihat_S) + \delta a(\Ihat_L) + (1 + \delta)$.

We model packing $J\lo \cup \Ihat_S$ as a linear program, denoted by $\LP(\Ihat)$,
that has at most $\ceildeltsq^{1/\delta}$ variables and $\ceildeltsq$ non-trivial constraints.
The optimum extreme point solution to $\LP(\Ihat)$, therefore,
has at most $\ceildeltsq$ positive entries,
so $\opt(J\lo \cup \Ihat_S) \le \opt(\LP(\Ihat)) + \ceildeltsq$.

We use techniques from Caprara~\cite{caprara2008} to obtain a
monotonic \dff{} $\eta$ from the optimal solution to the dual of $\LP(\Ihat)$.
For each item $i \in I$, we define $p(i) \defeq w(i)\eta(h(i))$
and prove that $p(I) \ge \opt(\LP(\Ihat))$.
By \cref{thm:dff-pack}, we get that $p(I) \le T_k^{d-1}\optdbp(I)$
and $a(\Ihat_L) \le T_k^{d-1}\optdbp(I)$.
Combining the above facts gives us an upper-bound on $\sopt_{\delta}(\Ihat)$
in terms of $\optdbp(I)$.
\end{proof}

\subsection{Subroutines}

\subsubsection{\texorpdfstring{$\guessShelves$}{guess-shelves}}

The algorithm $\guessShelves(\Icalhat, \delta)$ takes a set $\Icalhat$
of 2D itemsets and a constant $\delta \in (0, 1)$ as input.
We will design $\guessShelves$ so that it satisfies the following theorem.
\begin{theorem}
\label{thm:hgap:guess-shelves}
$\guessShelves(\Icalhat, \delta)$ returns all possible packings of empty shelves into at most
$|\Icalhat|$ bins such that each packing is structured for $(\flatten(\Icalhat), \delta)$.
$\guessShelves(\Icalhat, \delta)$ returns at most $T \defeq (N^{\ceildeltsq}+1)(n+1)^R$
packings, where $N \defeq |\flatten(\Icalhat)|$, $n \defeq |\Icalhat|$,
and $R \defeq \binom{\ceildeltsq + \ceil{1/\delta}-1}{\ceil{1/\delta}-1}
\le (1+\ceildeltsq)^{1/\delta}$. Its running time is $O(T)$.
\end{theorem}
$\guessShelves$ works by first guessing at most $\ceildeltsq$ distinct heights of shelves.
It then enumerates all configurations, i.e.,
different ways in which shelves can be packed into a bin.
It then guesses the configurations in a bin packing of the shelves.
$\guessShelves$ can be easily implemented using standard techniques.
For the sake of completeness, we give a more precise description of $\guessShelves$
and prove \cref{thm:hgap:guess-shelves} in \cref{sec:hgap:guess-shelves}.

\subsubsection{\texorpdfstring{$\chooseAndPack$}{choose-and-pack}}

$\chooseAndPack(\Icalhat, P, \delta)$ takes as input a set $\Icalhat$ of 2D itemsets,
a constant $\delta \in (0, 1)$, and a bin packing $P$ of empty shelves that is
structured for $(\flatten(\Icalhat), \delta)$.
It tries to pack an assortment of $\Icalhat$ into the shelves in $P$.

$\chooseAndPack$ works by rounding up the width of all
$\delta$-large items in $\Icalhat$ to a multiple of $1/n$.
This would increase the number of shelves required by 1, so it adds another empty shelf.
It then uses dynamic programming to pack an assortment into the shelves,
such that the area of the chosen $\delta$-small items is minimum.
This is done by maintaining a dynamic programming table that keeps track of
the number of itemsets considered so far and
the remaining space in shelves of each type.
If it is not possible to pack the items into the shelves,
then $\chooseAndPack$ outputs \Null.
In \cref{sec:hgap:choose-and-pack}, we give the details of this algorithm
and formally prove the following theorems:

\begin{restatable}{theorem}{rthmCAPNotNull}
\label{thm:hgap:cap-not-null}
If there exists an assortment $\Khat$ of $\Icalhat$ having
a structured $\delta$-fractional bin packing $P$,
then $\chooseAndPack(\Icalhat, P, \delta)$ does not output \Null.
\end{restatable}
\begin{restatable}{theorem}{rthmCAPCorrect}
\label{thm:hgap:cap-correct}
If the output of $\chooseAndPack(\Icalhat, P, \delta)$ is not \Null,
then the output $\Pbar$ is a shelf-based $\delta$-fractional packing of some assortment
of $\Icalhat$ such that $|\Pbar| \le |P| + 1$
and the distinct shelf heights in $\Pbar$ are the same as that in $P$.
\end{restatable}
\begin{restatable}{theorem}{rthmCAPTime}
\label{thm:hgap:cap-time}
$\chooseAndPack(\Icalhat, P, \delta)$ runs in $O(Nn^{2\ceildeltsq})$ time.
Here $N \defeq |\flatten(\Icalhat)|$, $n \defeq |\Icalhat|$.
\end{restatable}

\subsubsection{\texorpdfstring{$\inflate$}{inflate}}

For a set $I$ of $d$D items, $\inflate$ is an algorithm that
converts a shelf-based packing of $\round(I)$ into a packing of $I$
having roughly the same number of bins.

For a $d$D item $i$, $\btype(i)$ (called \emph{base type}) is defined to be
a $(d-1)$-dimensional vector whose $j\Th$ component is $\type_k(\ell_j(i))$.
Roughly, $\inflate(P)$ works as follows:
It first slightly modifies the packing $P$ so that
items of different base types are in different shelves and
$\delta$-small items are no longer sliced using horizontal cuts.
Then it converts each 2D shelf to a $d$D shelf of the same height using $\hdhkunitHyp$
(a $d$D shelf is a cuboid where the first $d-1$ dimensions are equal to 1).

In \cref{sec:hgap:inflate}, we formally describe $\inflate$
and prove the following theorem.
\begin{theorem}
\label{thm:hgap:inflate}
Let $I$ be a set of $d$D items having $Q$ distinct base types
(there can be at most $k^{d-1}$ distinct base types, so $Q \le k^{d-1}$).
Let $P$ be a shelf-based $\delta$-fractional packing of $\round(I)$
where shelves have $t$ distinct heights.
Then $\inflate(P)$ returns a packing of $I$ into less than
$|P|/(1-\delta) + t(Q-1) + 1 + \delta Q/(1-\delta)$ bins
in $O(|I|d\log|I|)$ time.
\end{theorem}

Now that we have mentioned the guarantees of all the subroutines used by $\hgapk$,
we can prove the correctness and running time of $\hgapk$.

\subsection{Correctness and Running Time of \texorpdfstring{$\hgapk$}{HGaP}}

\begin{theorem}
\label{thm:hgap:hgap}
The number of bins used by $\hgapk(\Ical, \eps)$ to pack $\Ical$ is less than
\[ T_k^{d-1}(1+\eps)\optdmcbp(\Ical)
    + \ceil{\left(\frac{2}{\eps}+1\right)^2}\left(Q + \frac{\eps}{2}\right)
    + 3 + (Q+3)\frac{\eps}{2}. \]
Here $Q \le k^{d-1}$ is the number of distinct base types in $\flatten(\Ical)$.
\end{theorem}
\begin{proof}
Let $K^*$ be the assortment in an optimal bin packing of $\Ical$.
Let $\Khat^* = \round(K^*)$.
Let $P^*$ be the optimal structured $\delta$-fractional bin packing of $\Khat^*$.
Then $|P^*| = \sopt_{\delta}(\Khat^*)$ by the definition of $\sopt$.
By \thmdepcref{thm:hgap:guess-shelves}{}, $P^* \in \guessShelves(\Icalhat, \delta)$.
Let $\Pbar^* = \chooseAndPack(\Icalhat, P^*, \delta)$.
By \thmdepcref{thm:hgap:cap-not-null}{}, $\Pbar^*$ is not \Null{}.
By \thmdepcref{thm:hgap:cap-correct}{}, $P_{\best}$ is
structured for $(\flatten(\Icalhat), \delta)$
and $|P_{\best}| \le |\Pbar^*| \le \sopt_{\delta}(\Khat^*) + 1$.

By \thmdepcref{thm:hgap:inflate}{}, we get that
\[ |\inflate(P_{\best})| < \frac{\sopt_{\delta}(\Khat^*)}{1-\delta}
    + \ceil{\frac{1}{\delta^2}}(Q-1) + 1 + \frac{\delta Q + 1}{1-\delta}. \]
By \thmdepcref{thm:hgap:struct}{} (structural theorem)
and using $\optdbp(K^*) = \optdmcbp(\Ical)$, we get
\[ \sopt_{\delta}(\Khat^*)
< T_k^{d-1}(1+\delta)\optdmcbp(\Ical) + \ceildeltsq + 1 + \delta. \]
Therefore, $|\inflate(P_{\best})|$ is less than
\begin{align*}
&T_k^{d-1}\frac{1+\delta}{1-\delta}\optdmcbp(\Ical)
    + \ceil{\frac{1}{\delta^2}}\left(Q + \frac{\delta}{1-\delta}\right)
    + 3 + \frac{\delta(3+Q)}{1-\delta}
\\ &= T_k^{d-1}(1+\eps)\optdmcbp(\Ical)
    + \ceil{\left(\frac{2}{\eps}+1\right)^2}\left(Q + \frac{\eps}{2}\right)
    + 3 + (Q+3)\frac{\eps}{2}.
\qedhere \end{align*}
\end{proof}
\begin{theorem}
\label{thm:hgap:hgap-time}
$\hgapk(\Ical, \eps)$ runs in time $O(N^{1+\ceildeltsq}n^{R+2\ceildeltsq}+Nd+nd\log n)$, where
$n \defeq |\Icalhat|$, $N \defeq |\flatten(\Icalhat)|$, $\delta \defeq \eps/(2+\eps)$
and $R \defeq \binom{\ceildeltsq + \ceil{1/\delta}-1}{\ceil{1/\delta}-1}
\le (1+\ceildeltsq)^{1/\delta}$.
\end{theorem}
\begin{proof}
Follows from \thmdepcref{thm:hgap:guess-shelves,thm:hgap:cap-time,thm:hgap:inflate}{}.
\end{proof}
\Cref{sec:hgap:improve-time} gives hints on improving the running time of $\hgapk$.

\subsection{\safed{}BP with Rotations}

We can solve the rotational version of $d$BP by reducing it to $d$MCBP
and using the $\hgapk$ algorithm.
Since each item can have up to $d!$ orientations,
the running time is polynomial in $nd!$, which is large when $d$ is large.
But we can do better for some special cases.

When the bin has the same length in each dimension, then for any item $i$,
$w(i) \defeq \prod_{j=1}^{d-1} f_k(\ell_j(i))$
is invariant to permuting the first $d-1$ dimensions.
In the first step of $\hgapk$, we replace each $d$D item $i$ by a rectangle
of width $w(i)$ and height $\ell_d(i)$.
So, instead of considering all $d!$ orientations,
we just need to consider at most $d$ different orientations,
where each orientation has a different length in the $d\Th$ dimension.

Suppose there are no orientation constraints, i.e., all $d!$ orientations of each item are allowed.
Let $L_j$ be the length of the bin in the $j\Th$ dimension, for each $j \in [d]$.
Analogous to the trick in \cref{sec:fhk-rot}, we first fix the $d\Th$ dimension of the item
and then optimally permute the first $d-1$ dimensions using a max-weight bipartite matching algorithm.
Hence, we need to consider only $d$ orientations instead of $d!$.

\appendix
\section{Details of the \texorpdfstring{$\hgapk$}{HGaP} Algorithm}
\label{sec:hgap-extra}

This section gives details of the subroutines used by $\hgapk$.
It also proves the theorems claimed in \cref{sec:hgap}.

\subsection{Preliminaries}

\begin{definition}
\label{defn:pred}
Let $I_1$ and $I_2$ be sets of 1D items.
Then $I_1$ is defined to be a predecessor of $I_2$ (denoted as $I_1 \preceq I_2$)
iff there exists a one-to-one mapping $\pi: I_1 \mapsto I_2$ such that
$\forall i \in I_1, \size(i) \le \size(\pi(i))$.
\end{definition}
\begin{observation}
\label{obs:pred-pack}
Let $I_1 \preceq I_2$ where $\pi$ is the corresponding mapping.
Then we can obtain a packing of $I_1$ from a packing of $I_2$,
by packing each item $i \in I_1$ in the place of $\pi(i)$.
Hence, $\opt(I_1) \le \opt(I_2)$.
\end{observation}

\begin{definition}[Canonical shelving]
\label{defn:hgap:can-shelv}
Let $I$ be a set of rectangles.
Order the items in $I$ in non-increasing order of height
(break ties arbitrarily but deterministically)
and greedily pack them into tight shelves,
slicing items using vertical cuts if necessary.
The set of shelves thus obtained is called the \emph{canonical shelving} of $I$,
and is denoted by $\canShelv(I)$.
(The canonical shelving is unique because ties are broken deterministically.)
\end{definition}

See \cref{fig:can-shelv-example} for an example of canonical shelving.

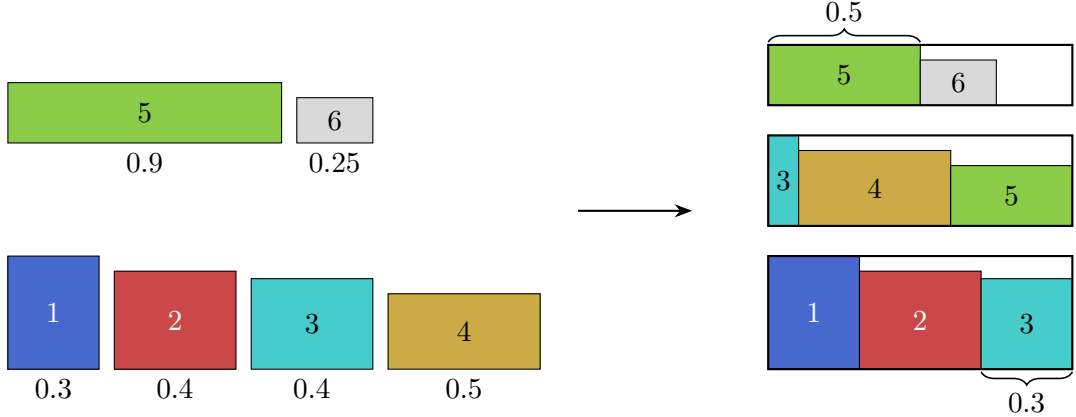
\begin{figure}[!htb]
\centering
\begin{tikzpicture}[
item/.style = {draw,fill={black!15}},
bin/.style={draw,thick},
myarrow/.style={->,>={Stealth},thick},
mybrace/.style = {decoration={brace,mirror,raise=1pt,amplitude=5pt},semithick,decorate},
]
\definecolor{myblue}{HTML}{4768CC}
\definecolor{myred}{HTML}{CC4747}
\definecolor{mycyan}{HTML}{47CCCC}
\definecolor{mygreen}{HTML}{8ACC47}
\definecolor{myyellow}{HTML}{CCAB47}
\begin{scope}
\path[item,fill=myblue] (0,0) rectangle +(1.2,1.5) node[text=white,pos=0.5] {1};
\node[anchor=north] at (0.6,0) {0.3};
\path[item,fill=myred] (1.4,0) rectangle +(1.6,1.3) node[text=white,pos=0.5] {2};
\node[anchor=north] at (2.2,0) {0.4};
\path[item,fill=mycyan] (3.2,0) rectangle +(1.6,1.2) node[pos=0.5] {3};
\node[anchor=north] at (4,0) {0.4};
\path[item,fill=myyellow] (5,0) rectangle +(2,1) node[pos=0.5] {4};
\node[anchor=north] at (6,0) {0.5};
\path[item,fill=mygreen] (0,3) rectangle +(3.6,0.8) node[pos=0.5] {5};
\node[anchor=north] at (1.8,3) {0.9};
\path[item] (3.8,3) rectangle +(1,0.6) node[pos=0.5] {6};
\node[anchor=north] at (4.3,3) {0.25};
\end{scope}
\draw[myarrow] (7.5,2.1) -- (9,2.1);
\begin{scope}[xshift=10cm]
\begin{scope}
\path[item,fill=myblue] (0.0,0) rectangle +(1.2,1.5) node[text=white,pos=0.5] {1};
\path[item,fill=myred] (1.2,0) rectangle +(1.6,1.3) node[text=white,pos=0.5] {2};
\path[item,fill=mycyan] (2.8,0) rectangle +(1.2,1.2) node[pos=0.5] {3};
\path[bin] (0,0) rectangle +(4,1.5);
\draw[mybrace] (2.8,0) -- node[below=5pt] {0.3} (4,0);
\end{scope}
\begin{scope}[yshift=1.9cm]
\path[item,fill=mycyan] (0,0) rectangle +(0.4,1.2) node[pos=0.5] {3};
\path[item,fill=myyellow] (0.4,0) rectangle +(2,1) node[pos=0.5] {4};
\path[item,fill=mygreen] (2.4,0) rectangle +(1.6,0.8) node[pos=0.5] {5};
\path[bin] (0,0) rectangle +(4,1.2);
\end{scope}
\begin{scope}[yshift=3.5cm]
\path[item,fill=mygreen] (0,0) rectangle +(2,0.8) node[pos=0.5] {5};
\path[item] (2,0) rectangle +(1,0.6) node[pos=0.5] {6};
\path[bin] (0,0) rectangle +(4,0.8);
\draw[mybrace] (2,0.8) -- node[above=5pt] {0.5} (0,0.8);
\end{scope}
\end{scope}
\end{tikzpicture}

\caption{Six items and their canonical shelving into three tight shelves of width 1.
The items are numbered by decreasing order of height.
Each item has its width mentioned below it.
Item 3 was sliced into two items of widths 0.3 and 0.1.
Item 5 was sliced into two items of widths 0.4 and 0.5.}
\label{fig:can-shelv-example}
\end{figure}

Suppose a set $I$ of rectangular items is packed into a set $J$ of shelves.
Then we can interpret $J$ as a 1BP instance where
the height of each shelf is the size of the corresponding 1D item.
We will now prove that the canonical shelving is optimal,
i.e., any shelf-based bin packing of items can be obtained
by first computing the canonical shelving and then
packing the shelves into bins like a 1BP instance.

\begin{lemma}
\label{thm:hgap:can-shelv-pred}
Let $I$ be a set of rectangles packed inside shelves $J$.
Let $J^* \defeq \canShelv(I)$. Then $J^* \preceq J$.
\end{lemma}
\begin{proof}
We say that a shelf is full if the total width of items in a shelf is 1.
Arrange the shelves $J$ in non-increasing order of height,
and arrange the items $I$ in non-increasing order of height.
Then try to pack $I$ into $J$ using the following greedy algorithm:
For each item $i$, pack the largest possible slice of $i$ into
the first non-full shelf and pack the remaining slice (if any) in the next shelf.
If this greedy algorithm succeeds, then within each shelf of $J$,
there is a shelf of $J^*$, so $J^* \preceq J$.
We will now prove that this greedy algorithm always succeeds.

For the sake of proof by contradiction, assume that the greedy algorithm failed, i.e.,
for an item (or slice) $i$ there was a non-full shelf $S$ but $h(i) > h(S)$.
Let $I'$ be the items (and slices) packed before $i$ and $J'$ be the shelves before $S$.
Therefore, $w(I') = |J'|$.

All items in $I'$ have height at least $h(i)$,
so all shelves in $J'$ have height at least $h(i)$.
All shelves after $J'$ have height less than $h(i)$.
Therefore, $J'$ is exactly the set of shelves of height at least $h(i)$.

In the packing $P$, $I' \cup \{i\}$ can only be packed into shelves of height
at least $h(i)$, so $w(I') + w(i) \le |J'|$. But this contradicts $w(I') = |J'|$.
Therefore, the greedy algorithm cannot fail.
\end{proof}

\begin{lemma}
\label{thm:nos-sum}
Consider the inequality $x_1 + x_2 + \ldots + x_n \le s$,
where for each $j \in [n]$, $x_j \in \mathbb{Z}_{\ge 0}$.
Let $N$ be the number of solutions to this inequality.
Then $N = \binom{s+n}{n} \le (s+1)^n$.
\end{lemma}
\begin{proof}
The proof of $N = \binom{s+n}{n}$ is a standard result in combinatorics.

To prove $N \le (s+1)^n$, note that we can choose each
$x_j \in \{0, 1, \ldots, s\}$ independently.
\end{proof}

\subsection{Structural Theorem}
\label{sec:hgap:struct}

Let $I$ be a set of $d$D items.
Let $\Ihat \defeq \round(I)$.
Let $\delta \in (0, 1)$ be a constant.
Let $\Ihat_L \defeq \{i \in \Ihat: h(i) > \delta\}$ and $\Ihat_S \defeq \Ihat - \Ihat_L$.
Let $J \defeq \canShelv(\Ihat_L)$. Let $m \defeq |J|$, i.e., $J$ contains $m$ shelves.
We can interpret $\Ihat_S$ as a single sliceable 1D item of size $a(\Ihat_S)$.

We will show the existence of a structured $\delta$-fractional packing of $\Ihat$
into at most $T_k^{d-1}(1+\delta)\optdbp(I) + \ceildeltsq + 1 + \delta$ bins.
This would prove \cref{thm:hgap:struct}.

\begin{definition}[Linear grouping~\cite{bp-aptas}]
\label{defn:hgap:lingroup}
Arrange the 1D items $J$ in non-increasing order of size and number them from 1 to $m$.
Let $q \defeq \floor{\delta\size(J)} + 1$.
Let $J_1$ be the first $q$ items, $J_2$ be the next $q$ items, and so on.
$J_j$ is called the $j\Th$ \emph{linear group} of $J$.
This gives us $t \defeq \ceil{m/q}$ linear groups.
Note that the last group, $J_t$, may have less than $q$ items.

Let $h_j$ be the size of the first item in $J_j$. Let $h_{t+1} \defeq 0$.
For $j \in [t-1]$, let $J_j\lo$ be the items obtained by
decreasing the height of items in $J_j$ to $h_{j+1}$.
For $j \in [t]$, let $J_j\hi$ be the items obtained by
increasing the height of items in $J_j$ to $h_j$.

Let ${J\lo \defeq \bigcup_{j=1}^{t-1} J_j\lo}$
and ${J\hi \defeq \bigcup_{j=1}^t J_j\hi}$.
We call $J\lo$ a down-rounding of $J$ and $J\hi$ an up-rounding of $J$.
\end{definition}

\begin{lemma}
\label{thm:hgap:n-pivots}
$t \le \ceildeltsq$.
\end{lemma}
\begin{proof}
Since each shelf in $J$ has height more than $\delta$, $\size(J) > |J|\delta$.
\[ t \defeq \ceil{\frac{|J|}{\floor{\delta\size(J)}+1}}
\le \ceil{\frac{\size(J)/\delta}{\delta\size(J)}}
= \ceil{\frac{1}{\delta^2}}.
\qedhere \]
\end{proof}

\begin{lemma}
\label{thm:hgap:pred-chain}
$J\lo \preceq J \preceq J\hi \preceq J\lo \cup J_1\hi$.
\end{lemma}
\begin{proof}
It is trivial to see that $J\lo \preceq J \preceq J\hi$.
For $j \in [t-1]$, all (1D) items in both $J_j\lo$ and $J_{j+1}\hi$ have height $h_{j+1}$,
and $|J_{j+1}| \le q = |J_j|$. Therefore, $J_{j+1}\hi \preceq J_j\lo$ and hence
\[ J\hi = J_1\hi \cup \bigcup_{j=1}^{t-1} J_{j+1}\hi
\preceq J_1\hi \cup \bigcup_{j=1}^{t-1} J_j\lo = J_1\hi \cup J\lo.
\qedhere \]
\end{proof}

\begin{lemma}
\label{thm:hgap:can-shelv-size}
$\size(J) < 1 + a(\Ihat_L)$.
\end{lemma}
\begin{proof}
In the canonical shelving of $\Ihat_L$, let $S_j$ be the $j\Th$ shelf.
Let $h(S_j)$ be the height of $S_j$.
Let $a(S_j)$ be the total area of the items in $S_j$.
Since the shelves are tight, items in $S_j$ have height at least $h(S_{j+1})$.
So, $a(S_j) \ge h(S_{j+1})$ and
\[ \size(J) = \sum_{j=1}^{|J|} h(S_j) \le 1 + \sum_{j=1}^{|J|-1} h(S_{j+1})
\le 1 + \sum_{j=1}^{|J|-1} a(S_j) < 1 + a(\Ihat_L).
\qedhere \]
\end{proof}

\begin{lemma}
\label{thm:hgap:sopt-le-optlo}
$\sopt_{\delta}(\Ihat) < \opt(J\lo \cup \Ihat_S) + \delta a(\Ihat_L) + (1 + \delta)$.
\end{lemma}
\begin{proof}
By the definition of $\canShelv$, $\Ihat_L$ can be packed into $J$.
By \thmdepcref{thm:hgap:pred-chain}{}, $J \preceq J\hi$, so $\Ihat_L$ can be packed into $J\hi$.
By \thmdepcref{thm:hgap:n-pivots}{}, the number of distinct sizes in $J\hi$
is at most $\ceildeltsq$.
So, the optimal 1D bin packing of $J\hi \cup \Ihat_S$ will
give us a structured $\delta$-fractional bin packing of $\Ihat$.
Hence, $\sopt_{\delta}(\Ihat) \le \opt(J\hi \cup \Ihat_S)$.

By \thmdepcref{thm:hgap:pred-chain,obs:pred-pack}{} we get
\[ \opt(J\hi \cup \Ihat_S) \le \opt(J\lo \cup J_1\hi \cup \Ihat_S)
\le \opt(J\lo \cup \Ihat_S) + \opt(J_1\hi). \]
By \thmdepcref{thm:hgap:can-shelv-size}{},
\[ \opt(J_1\hi) \le |J_1\hi| \le q \le 1 + \delta\size(J)
< 1 + \delta(1 + a(\Ihat_L)).
\qedhere \]
\end{proof}

\subsubsection{LP for Packing \texorpdfstring{$J\lo \cup \Ihat_S$}{J\^{}lo + I\^{}\_S}}

We will formulate an integer linear program for bin packing $J\lo \cup \Ihat_S$.

Let $C \in \mathbb{Z}_{\ge 0}^{t-1}$ such that $h_C \defeq \sum_{j=1}^{t-1} C_jh_{j+1} \le 1$.
Then $C$ is called a \config{}.
$C$ represents a set of 1D items that can be packed into a bin
and where $C_j$ items are from $J_j\lo$.
Let $\Ccal$ be the set of all \config{}s.
We can pack at most $\ceil{1/\delta}-1$ items into a bin because $h_t > \delta$.
By \cref{thm:nos-sum}, we get
$|\Ccal| \le \binom{\ceil{1/\delta}-1+t-1}{t-1} \le \ceildeltsq^{1/\delta}$.

Let $x_C$ be the number of bins packed according to \config{} $C$.
Bin packing $J\lo \cup \Ihat_S$ is equivalent to finding the optimal integer solution
to the following linear program, which we denote as $\LP(\Ihat)$.
\[ \begin{array}{*3{>{\displaystyle}l}}
\min_{x \in \mathbb{R}^{|\Ccal|}} & \sum_{C \in \Ccal} x_C
\\[1.5em] \textrm{where }
    & \sum_{C \in \Ccal} C_j x_C \ge q & \forall j \in [t-1]
\\[1.5em] & \sum_{C \in \Ccal} (1-h_C)x_C \ge a(\Ihat_S)
\\[1em] & x_C \ge 0 & \forall C \in \Ccal
\end{array} \]
Here the first set of constraints say that for each $j \in [t-1]$,
all of the $q \defeq \floor{\delta\size(J)}+1$ shelves $J\lo_j$
should be covered by the configurations in $x$.
The second constraint says that we should be able to pack $a(\Ihat_S)$
into the non-shelf space in the bins.
\begin{lemma}
\label{thm:hgap:opt-to-lp}
$\opt(J\lo \cup \Ihat_S) \le \opt(\LP(\Ihat)) + t$.
\end{lemma}
\begin{proof}
Let $x^*$ be an optimal extreme-point solution to $\LP(\Ihat)$.
By rank-lemma, $x^*$ has at most $t$ non-zero entries.
Let $\xhat$ be a vector where $\xhat_C \defeq \ceil{x_C^*}$.
Then $\xhat$ is an integral solution to $\LP(\Ihat)$ and
$\sum_C \xhat_C < t + \sum_C x_C^* = \opt(\LP(\Ihat)) + t$.
\end{proof}
The dual of $\LP(\Ihat)$, denoted by $\DLP(\Ihat)$, is
\begin{align*}
\max_{y \in \mathbb{R}^{t-1}, z \in \mathbb{R}}
    & a(\Ihat_S)z + q \sum_{j=1}^{t-1} y_j
\\ \textrm{ where } & \sum_{j=1}^{t-1} C_j y_j + (1-h_C)z \le 1 \quad \forall C \in \Ccal
\\ & z \ge 0 \textrm{ and } y_j \ge 0 \quad \forall j \in [t-1]
\end{align*}

\subsubsection{Weighting Function for a Feasible Solution to
\texorpdfstring{$\DLP(\Ihat)$}{DLP(I\^{})}}

We will now see how to obtain a monotonic weighting function
$\eta: [0, 1] \mapsto [0, 1]$ from a feasible solution to $\DLP(\Ihat)$.
To do this, we adapt techniques from Caprara's analysis of $\hdhk$~\cite{caprara2008}.
Such a weighting function will help us upper-bound $\opt(\LP(\Ihat))$
in terms of $\optdbp(I)$.

We first describe a transformation that helps us convert any feasible
solution of $\DLP(\Ihat)$ to a feasible solution that is \emph{monotonic}.
We then show how to obtain a weighting function from this monotonic solution.

\begin{transformation}
\label{trn:hgap:half-mono}
Let $(y, z)$ be a feasible solution to $\DLP(\Ihat)$. Let $s \in [t-1]$.
Define $y_t \defeq 0$ and $h_{t+1} \defeq 0$.
Then change $y_s$ to $\max(y_s, y_{s+1} + (h_{s+1} - h_{s+2})z)$.
\end{transformation}
\begin{lemma}
\label{thm:hgap:half-mono-feas}
Let $(y, z)$ be a feasible solution to $\DLP(\Ihat)$.
Let $(\yhat, z)$ be the new solution obtained by applying \cref{trn:hgap:half-mono}
with parameter $s \in [t-1]$. Then $(\yhat, z)$ is feasible for $\DLP(\Ihat)$.
\end{lemma}
\begin{proof}
For a \config{} $C$, let $f(C, y, z) \defeq C^Ty + (1-h_C)z$,
where $C^Ty \defeq \sum_{j=1}^{t-1} C_jy_j$.
Since $(y, z)$ is feasible for $\DLP(\Ihat)$, $f(C, y, z) \le 1$.
As per \cref{trn:hgap:half-mono},
\[ \yhat_j \defeq \begin{cases} \max(y_s, y_{s+1} + (h_{s+1} - h_{s+2})z) & j = s
\\ y_j & j \neq s \end{cases}. \]
If $y_s \ge y_{s+1} + (h_{s+1} - h_{s+2})z$, then $\yhat = y$,
so $(\yhat, z)$ would be feasible for $\DLP(\Ihat)$.
So now assume that $y_s < y_{s+1} + (h_{s+1} - h_{s+2})z$.

Let $C$ be a \config{}. Define $C_t \defeq 0$. Let
\[ \Chat_j \defeq \begin{cases} 0 & j = s \\ C_s + C_{s+1} & j = s+1
\\ C_j & \textrm{otherwise} \end{cases}. \]
Then, $C^T\yhat - \Chat^Ty
= C_s\yhat_s + C_{s+1}\yhat_{s+1} - \Chat_sy_s - \Chat_{s+1}y_{s+1}
= C_s(h_{s+1} - h_{s+2}) z$.

Also, $h_{\Chat} - h_C
= \Chat_sh_{s+1} + \Chat_{s+1}h_{s+2} - C_sh_{s+1} - C_{s+1}h_{s+2}
= - C_s(h_{s+1} - h_{s+2})$.

Since $h_{\Chat} \le h_C \le 1$, $\Chat$ is a \config{}.
\begin{align*}
f(C, \yhat, z) &= C^T\yhat + (1-h_C)z
\\ &= (\Chat^Ty + C_s(h_{s+1} - h_{s+2})z)
    + (1 - h_{\Chat} - C_s(h_{s+1} - h_{s+2}))z
\\ &= f(\Chat, y, z) \le 1.
\end{align*}
Therefore, $(\yhat, z)$ is feasible for $\DLP(\Ihat)$.
\end{proof}

\begin{definition}
Let $(y, z)$ be a feasible solution to $\DLP(\Ihat)$. Let
\[ \yhat_j \defeq \begin{cases}
\max(y_{t-1}, zh_t) & j = t-1
\\ \max(y_j, \yhat_{j+1} + (h_{j+1} - h_{j+2})z) & j < t-1
\end{cases}. \]
Then $(\yhat, z)$ is called the monotonization of $(y, z)$.
\end{definition}
\begin{lemma}
\label{thm:hgap:mono-feas}
Let $(y, z)$ be a feasible solution to $\DLP(\Ihat)$.
Let $(\yhat, z)$ be the monotonization of $(y, z)$.
Then $(\yhat, z)$ is a feasible solution to $\DLP(\Ihat)$.
\end{lemma}
\begin{proof}
$(\yhat, z)$ can be obtained by multiple applications of \cref{trn:hgap:half-mono}:
first with $s = t-1$, then with $s = t-2$, and so on till $s = 1$.
Then by \thmdepcref{thm:hgap:half-mono-feas}{}, $(\yhat, z)$ is feasible for $\DLP(\Ihat)$.
\end{proof}

Let $(y^*, z^*)$ be an optimal solution to $\DLP(\Ihat)$.
Let $(\yhat, z^*)$ be the monotonization of $(y^*, z^*)$.
Then define the function $\eta: [0, 1] \mapsto [0, 1]$ as
\[ \eta(x) \defeq \begin{cases}
\yhat_1 & \textrm{if } x \in [h_2, 1]
\\ \yhat_j & \textrm{if } x \in [h_{j+1}, h_j), \textrm{ for } 2 \le j \le t-1
\\ xz^* & \textrm{if } x < h_t
\end{cases}. \]

\begin{lemma}
\label{thm:hgap:eta-dff}
$\eta$ is a monotonic \dff{}.
\end{lemma}
\begin{proof}
$\eta$ is monotonic by the definition of monotonization.

Let $X \subseteq (0, 1]$ be a finite set such that $\Sum(X) \le 1$.
Let $X_0 \defeq X \cap [0, h_t)$,
let $X_1 \defeq X \cap [h_2, 1]$
and for $2 \le j \le t-1$, let $X_j \defeq X \cap [h_{j+1}, h_j)$.
Let $C \in \mathbb{Z}^{t-1}_{\ge 0}$ such that $C_j \defeq |X_j|$.
Let $h_C \defeq \sum_{j=1}^{t-1} C_jh_{j+1}$.
\begin{align*}
1 \ge \Sum(X) &= \Sum(X_0) + \sum_{j=1}^{t-1} \Sum(X_j)
\\ &\ge \Sum(X_0) + \sum_{j=1}^{t-1} C_jh_{j+1}
\tag{for $j \ge 1$, each element in $X_j$ is at least $h_{j+1}$}
\\ &= \Sum(X_0) + h_C.
\end{align*}
Since $h_C \le 1 - \Sum(X_0) \le 1$, $C$ is a \config{}. Therefore,
\begin{align*}
\sum_{x \in X} \eta(x)
&= \sum_{j=0}^{t-1} \sum_{x \in X_j} \eta(x)
= z^*\Sum(X_0) + \sum_{j=1}^{t-1} C_j\yhat_j
\tag{by definition of $\eta$}
\\ &\le (1-h_C)z^* + C^T\yhat
\tag{$h_C \le 1 - \Sum(X_0)$}
\\ &\le 1.
\tag{$C$ is a \config{} and $(\yhat, z^*)$ is feasible for $\DLP(\Ihat)$
    by \thmdepcref{thm:hgap:mono-feas}{}}
\end{align*}
\end{proof}

\begin{lemma}
\label{thm:hgap:p-vs-opt}
For $i \in I$, let $p(i) \defeq \eta(h(i))w(i)$.
Then $\opt(\LP(\Ihat)) \le p(I) \le T_k^{d-1}\optdbp(I)$.
\end{lemma}
\begin{proof}
Let $(y^*, z^*)$ be an optimal solution to $\DLP(\Ihat)$.
Let $(\yhat, z^*)$ be the monotonization of $(y^*, z^*)$.

In the canonical shelving of $I$, suppose a rectangular item $i$ (or a slice thereof)
lies in shelf $S$ where $S \in J_j$.
Then $h(i) \in [h_{j+1}, h_j]$, where $h_{t+1} \defeq 0$.
This is because shelves in $J \defeq \canShelv(\Ihat)$ are tight.
If $j = 1$, then $\eta(h(i)) = \yhat_1 \ge y^*_1$.
If $2 \le j \le t-1$, then $\eta(h(i)) \in \{\yhat_{j-1}, \yhat_j\} \ge \yhat_j \ge y^*_j$.
We are guaranteed that for $j \in [t-1]$, and each shelf $S \in J_j$, $w(S) = 1$.
\begin{align*}
p(I) &= \sum_{j=1}^t \sum_{S \in J_j} \sum_{i \in S} \eta(h(i))w(i)
    + \sum_{i \in \Ihat_S} \eta(h(i))w(i)
    \tag{by definition of $p$}
\\ &\ge \sum_{j=1}^{t-1} \sum_{S \in J_j} \sum_{i \in S} y^*_j w(i)
    + \sum_{i \in \Ihat_S} (h(i)z^*)w(i)
\tag{by definition of $\eta$}
\\ &= \sum_{j=1}^{t-1} y^*_j q + a(\Ihat_S)z^*
\tag{since $w(J_j) = q$ for $j \in [t-1]$}
\\ &= \opt(\DLP(\Ihat)).
\tag{$(y^*, z^*)$ is optimal for $\DLP(\Ihat)$}
\end{align*}
By strong duality of linear programs, $\opt(\LP(\Ihat)) = \opt(\DLP(\Ihat)) \le p(I)$.

Since $\eta$ and $H_k$ are \dff{}s (by \thmdepcref{thm:hgap:eta-dff}{}),
we get that $p(I) \le T_k^{d-1}\optdbp(I)$ by \thmdepcref{thm:dff-pack}{}.
\end{proof}

\rthmHgapStruct*
\begin{proof}
\begin{align*}
& a(\Ihat_L) \le a(\Ihat)
= \sum_{i \in I} \left(\ell_d(i)\prod_{j=1}^{d-1} f_k(\ell_j(i))\right)
\le T_k^{d-1}\optdbp(I).
\tag{by \thmdepcref{thm:dff-pack}{thm:hgap:struct}}
\end{align*}
\begin{align*}
\sopt_{\delta}(\Ihat) &< \opt(J\lo \cup \Ihat_S) + \delta a(\Ihat_L) + (1+\delta)
\tag{by \thmdepcref{thm:hgap:sopt-le-optlo}{thm:hgap:struct}}
\\ &\le \opt(\LP(\Ihat)) + \ceil{\frac{1}{\delta^2}} + \delta T_k^{d-1}\optdbp(I) + (1+\delta)
\tag{by \thmdepcref{thm:hgap:opt-to-lp,thm:hgap:n-pivots}{thm:hgap:struct}}
\\ &\le T_k^{d-1}(1+\delta)\optdbp(I) + \ceil{\frac{1}{\delta^2}} + 1 + \delta.
\tag{by \thmdepcref{thm:hgap:p-vs-opt}{thm:hgap:struct}}
\end{align*}
\end{proof}

\subsection{Guessing Shelves and Bins}
\label{sec:hgap:guess-shelves}

We want $\guessShelves(\Icalhat, \delta)$ to return all possible packings
of empty shelves into at most $n \defeq |\Icalhat|$ bins such that
each packing is structured for $(\flatten(\Icalhat), \delta)$.

Let $H = \{h(i): i \in \flatten(\Icalhat)\}$. Let $N \defeq |\flatten(\Icalhat)|$.
$\guessShelves(\Icalhat, \delta)$ starts by picking the distinct heights of shelves
by iterating over all subsets of $H$ of size at most $\ceildeltsq$.
The number of such subsets is at most $N^{\ceildeltsq}+1$.
Let $\Htild \defeq \{h_1, h_2, \ldots, h_t\}$ be one such guess, where $t \le \ceildeltsq$.
\WLoG, assume $h_1 > h_2 > \ldots > h_t > \delta$.

Next, $\guessShelves$ needs to decide the number of shelves of each height
and a packing of those shelves into bins.
Let $C \in \mathbb{Z}_{\ge 0}^t$ such that $h_C \defeq \sum_{j=1}^{t-1} C_jh_j \le 1$.
Then $C$ is called a \config{}.
$C$ represents a set of shelves that can be packed into a bin
and where $C_j$ shelves have height $h_j$.
Let $\Ccal$ be the set of all \config{}s.
We can pack at most $\ceil{1/\delta}-1$ items into a bin because $h_t > \delta$.
By \thmdepcref{thm:nos-sum}{thm:hgap:guess-shelves}, we get
\[ |\Ccal| \le \binom{\ceil{1/\delta}-1+t}{t}
\le \binom{\ceil{1/\delta}-1+\ceildeltsq}{\ceil{1/\delta}-1}
\le \left(\ceil{\frac{1}{\delta^2}}+1\right)^{1/\delta}. \]
There can be at most $n$ bins, and $\guessShelves$ has to decide the \config{} of each bin.
By \thmdepcref{thm:nos-sum}{thm:hgap:guess-shelves},
the number of ways of doing this is at most
$\binom{|\Ccal|+n}{|\Ccal|} \le (n+1)^{|\Ccal|}$.
Therefore, $\guessShelves$ computes all \config{}s and then
iterates over all $\binom{|\Ccal|+n}{|\Ccal|}$ combinations of these configs.

This completes the description of $\guessShelves$ and proves
\cref{thm:hgap:guess-shelves}.

\subsection{\texorpdfstring{$\chooseAndPack$}{choose-and-pack}}
\label{sec:hgap:choose-and-pack}

$\chooseAndPack(\Icalhat, P, \delta)$ takes as input a set $\Icalhat$ of 2D itemsets,
a packing $P$ of empty shelves into bins and constant $\delta \in (0, 1)$.
It tries to pack $\Icalhat$ into $P$ and one additional shelf.
Before we design $\chooseAndPack$, let us see how to handle a special case,
i.e., where $\Icalhat$ is \emph{simple}.
\begin{definition}
A set $\Icalhat$ of 2D itemsets is \emph{$\delta$-simple} iff
the width of each $\delta$-large item in $\flatten(\Icalhat)$
is a multiple of $1/|\Icalhat|$.
\end{definition}

Let $P$ be a bin packing of empty shelves.
Let $h_1 > h_2 > \ldots > h_t$ be the distinct heights of
the shelves in $P$, where $h_t > \delta$.
We will use dynamic programming to either pack a simple instance $\Icalhat$
into $P$ or claim that no assortment of $\Icalhat$ can be packed into $P$.
Call this algorithm $\simpleChooseAndPack(\Icalhat, P, \delta)$.

Let $\Icalhat \defeq \{I_1, I_2, \ldots, I_n\}$.
For $j \in \{0, 1, \ldots, n\}$, define $\Icalhat_j \defeq \{I_1, I_2, \ldots, I_j\}$,
i.e., $\Icalhat_j$ contains the first $j$ itemsets from $\Icalhat$.
Let $\vec{u} \defeq [u_1, u_2, \ldots, u_t] \in \{0, 1, \ldots, n^2\}^t$ be a vector.
Let $\Phi(j, \vec{u})$ be the set of all assortments of $\Icalhat_j$ that can be
packed into $t$ shelves, where the $r\Th$ shelf has height $h_r$ and width $u_r/n$.
For a set $K$ of items, define $\smallArea(K)$ as
the total area of $\delta$-small items in $K$.
Define $g(j, \vec{u}) \defeq \min_{K \in \Phi(j, \vec{u})} \smallArea(K)$.
If $\Phi(j, \vec{u}) = \emptyset$, then we let $g(j, \vec{u}) = \infty$.

We will show how to compute $g(j, \vec{u})$ for all $j \in \{0, 1, \ldots, n\}$
and all $\vec{u} \in \{0, 1, \ldots, n^2\}^t$ using dynamic programming.
Let there be $n_r$ shelves in $P$ having height $h_r$.
Then for $j = n$ and $u_r = n_r n$, $\Icalhat$ can be packed into $P$ iff
$g(j, \vec{u})$ is at most the area of non-shelf space in $P$.

Note that in any solution $K$ corresponding to $g(j, \vec{u})$,
we can assume without loss of generality that the item $i$ from $K \cap I_j$
is placed in the smallest shelves possible.
This is because we can always swap $i$ with the slices of items in those shelves.
This observation gives us the following recurrence relation for $g(j, \vec{u})$:
\begin{equation}
\label{eqn:g-rec}
g(j, \vec{u}) = \begin{cases}
\infty & \textrm{ if } u_j < 0 \textrm{ for some } j \in [t]
\\ 0 & \textrm{ if } n = 0 \textrm{ and } u_j \ge 0 \textrm{ for all } j \in [t]
\\ \min_{i \in I_j} \left(\begin{array}{ll}\smallArea(\{i\})
        \\ + \, g(j-1, \reduce(\vec{u}, i))\end{array} \right)
    & \textrm{ if } n > 0 \textrm{ and } u_j \ge 0 \textrm{ for all } j \in [t]
\end{cases}
\end{equation}
Here $\reduce(\vec{u}, i)$ is a vector obtained as follows:
If $i$ is $\delta$-small, then $\reduce(\vec{u}, i) \defeq \vec{u}$.
Otherwise, initialize $x$ to $w(i)$.
Let $p_i$ be the largest integer $r$ such that $h(i) \le h_r$.
For $r$ varying from $p_i$ to 2, subtract $\min(x, u_j)$ from $x$ and $u_j$.
Then subtract $x$ from $u_1$. The new value of $\vec{u}$ is defined to be
the output of $\reduce(\vec{u}, i)$.

The recurrence relation allows us to compute $g(j, \vec{u})$
for all $j$ and $\vec{u}$ using dynamic programming
in time $O(Nn^{2t})$ time, where $N \defeq |\flatten(\Icalhat)|$.
With a bit more work, we can also compute the corresponding assortment $K$, if one exists.
Therefore, $\simpleChooseAndPack(\Icalhat, P, \delta)$ computes a packing of $\Icalhat$
into $P$ if one exists, or returns \Null{} if no assortment of $\Icalhat$ can be packed into $P$.

Now we will look at the case where $\Icalhat$ is not $\delta$-simple.
Let $\Icalhat'$ be the instance obtained by rounding up the width
of each $\delta$-large item in $\Icalhat$ to a multiple of $1/n$, where $n \defeq |\Icalhat|$.
Let $\Pbar$ be the bin packing obtained by adding another bin to $P$
containing a single shelf of height $h_1$.
$\chooseAndPack(\Icalhat, P, \delta)$ computes $\Icalhat'$ and $\Pbar$
and returns the output of $\simpleChooseAndPack(\Icalhat', \Pbar, \delta)$.

\rthmCAPCorrect*
\begin{proof}
Follows from the definition of $\simpleChooseAndPack$.
\end{proof}

\rthmCAPNotNull*
\begin{proof}
Let $\Khat'$ be the items obtained by rounding up the width of each item in $\Khat$
to a multiple of $1/n$. Then $\Khat'$ is an assortment of $\Icalhat'$.
We will show that $\Khat'$ fits into $\Pbar$,
so $\simpleChooseAndPack(\Icalhat', \Pbar, \delta)$ will not output \Null.

Slice each item $i \in \Khat'$ into two pieces using a vertical cut such that
one piece has width equal to the original width of $i$ in $\Khat$,
and the other piece has width less than $1/n$.
This splits $\Khat'$ into sets $\Khat$ and $T$.
$T$ contains at most $n$ items, each of width less than $1/n$.
Therefore, we can pack $\Khat$ into $P$ and we can pack $T$ into
the newly-created shelf of height $h_1$.
Therefore, $\Khat'$ can be packed into $\Pbar$,
so $\simpleChooseAndPack(\Icalhat', \Pbar, \delta)$ won't output \Null.
\end{proof}

\rthmCAPTime*
\begin{proof}
The running time of $\chooseAndPack(\Icalhat, P, \delta)$
is dominated by computing $g(j, \vec{u})$ for all $j$ and $\vec{u}$,
which takes $O(Nn^{2t})$ time.
Since $P$ is structured for $(\Icalhat, \delta)$, the number of distinct shelves
in $P$, which is $t$, is at most $\ceildeltsq$.
\end{proof}

\subsection{\texorpdfstring{$\inflate$}{inflate}}
\label{sec:hgap:inflate}

Let $I$ be a set of $d$D items. Let $P$ be a shelf-based $\delta$-fractional
bin packing of $\Ihat \defeq \round(I)$ into $m$ bins.
Let there be $t$ distinct heights of shelves in $P$: $h_1 > h_2 > \ldots > h_t > \delta$.
We want to design an algorithm $\inflate(P)$ that returns a packing
of $I$ into approximately $|P|$ bins.

Define $\Ihat_L \defeq \{i \in \Ihat: h(i) > \delta\}$ and $\Ihat_S \defeq \Ihat - \Ihat_L$.
Let there be $Q$ distinct base types in $I$ (so $Q \le k^{d-1}$).

\subsubsection{Separating Base Types}
\label{sec:hgap:sep-btype}

We will now impose an additional constraint over $P$:
items in each shelf must have the same $\btype$.
This will be helpful later, when we will try to compute a packing of $d$D items $I$.

Separating base types of $\Ihat_S$ is easy, since we can slice them
in both directions. An analogy is to think of a mixture of
multiple immiscible liquids of different densities settling into equilibrium.
%

Let there be $n_j$ shelves of height $h_j$.
Let $\Ihat_j$ be the items packed into shelves of height $h_j$.
Therefore, $w(\Ihat_j) \le n_j$.
Let $\Ihat_{j,q} \subseteq \Ihat_j$ be the items of base type $q \in [Q]$.

For each $q$, pack $\Ihat_{j,q}$ into $\smallceil{w(\Ihat_{j,q})}$ shelves of height $h_j$
(slicing items if needed). For these newly-created shelves, define the $\btype$
of the shelf to be the $\btype$ of the items in it.
Let the number of newly-created shelves of height $h_j$ be $n_j'$. Then
\[ n_j' = \sum_{q=1}^Q \smallceil{w(\Ihat_{j,q})}
< \sum_{q=1}^Q w(\Ihat_{j,q}) + Q
\le n_j + Q \implies n_j' \le n_j + Q - 1. \]
$n_j$ of these shelves can be packed into existing bins in place of the old shelves.
The remaining $n_j' - n_j \le Q-1$ shelves can be packed on the base of new bins.

Therefore, by using at most $t(Q-1)$ new bins,
we can ensure that for every shelf,
all items in that shelf have the same $\btype$.
These new bins don't contain any items from $\Ihat_S$.
Call this new bin packing $P'$.
This transformation takes $O(|I|d\log|I|)$ time.

\subsubsection{Forbidding Horizontal Slicing}
\label{sec:hgap:forbid-hslice}

We will now use $P'$ to compute a shelf-based bin packing $P''$ of $\Ihat$ where
items in $\Ihat$ can be sliced using vertical cuts only.

Let $\Ihat_{q,S}$ be the items in $\Ihat_S$ of base type $q$.
Pack items $\Ihat_{q,S}$ into shelves using $\canShelvHyp$.
Suppose $\canShelv$ used $m_q$ shelves to pack $\Ihat_{q,S}$.
For $j \in [m_q]$, let $h_{q,j}$ be the height of the $j\Th$ shelf.
Let $H_q \defeq \sum_{j=1}^{m_q} h_{q,j}$ and $H \defeq \sum_{q=1}^Q H_q$.
Since for $j \in [m_q-1]$, all items in the $j\Th$ shelf have height at least $h_{q,j+1}$,
\[ a(\Ihat_{q,S}) > \sum_{j=1}^{m_q-1} h_{q,j+1} \ge H_q - h_{q,1} \ge H_q - \delta. \]
Therefore, $H < a(\Ihat_S) + Q\delta$.
Let $\Jhat_S$ be the set of these newly-created shelves.

Use Next-Fit to pack $\Jhat_S$ into the space used by $\Ihat_S$ in $P'$.
$\Ihat_S$ uses at most $m$ bins in $P'$ (recall that $m \defeq |P|$).
A height of less than $\delta$ will remain unpacked in each of those bins.
The total height occupied by $\Ihat_S$ in $P'$ is $a(\Ihat_S)$.
Therefore, Next-Fit will pack a height of more than $a(\Ihat_S) - \delta m$.

Some shelves in $\Jhat_S$ may still be unpacked.
Their total height will be less than
$H - (a(\Ihat_S) - \delta m) < \delta(Q + m)$.
We will pack these shelves into new bins using Next-Fit.
The number of new bins used is at most $\ceil{\delta(Q + m)/(1-\delta)}$.
Call this bin packing $P''$. The number of bins in $P''$
is at most $m' \defeq m + t(Q-1) + \ceil{\delta(Q+m)/(1-\delta)}$.

\subsubsection{Shelf-Based \safed{}D packing}
\label{sec:hgap:2d-to-dd}

We will now show how to convert the packing $P''$ of $\Ihat$ that uses $m'$ bins
into a packing of $I$ that uses $m'$ $d$D bins.

First, we repack the items into the shelves.
For each $q \in [Q]$, let $\Jhat_q$ be the set of shelves in $P''$ of $\btype$ $q$.
Let $\Ihat^{[q]}$ be the items packed into $\Jhat_q$.
Compute $\Jhat^*_q \defeq \canShelvHyp(\Ihat^{[q]})$ and pack the shelves
$\Jhat^*_q$ into $\Jhat_q$. This is possible by
\thmdepcref{thm:hgap:can-shelv-pred}{thm:hgap:inflate}.

This repacking gives us an ordering of shelves in $\Jhat_q$.
Number the shelves from 1 onwards.
All items have at most 2 slices. If an item has 2 slices, and one slice is packed
into shelf number $p$, then the other slice is packed into shelf number $p+1$.
The slice in shelf $p$ is called the leading slice.
Every shelf has at most one leading slice.

Let $S_j$ be the $j\Th$ shelf of $\Jhat_q$.
Let $R_j$ be the set of unsliced items in $S_j$
and the item whose leading slice is in $S_j$.
Order the items in $R_j$ arbitrarily, except that the sliced item, if any, should be last.
Then $w(R_j - \last(R_j)) < 1$.
So, we can use $\hdhkunitHyp^{[q]}(R_j)$ to pack $R_j$ into a $(d-1)$D bin.
This $(d-1)$D bin gives us a $d$D shelf whose height is the same as that of $S_j$.
On repeating this process for all shelves in $\Jhat_q$ and for all $q \in [Q]$,
we get a packing of $I$ into shelves.
Since each $d$D shelf corresponds to a shelf in $P''$ of the same height,
we can pack these $d$D shelves into bins in the same way as $P''$.
This gives us a bin packing of $I$ into $m'$ bins.

\subsubsection{The Algorithm}

\Cref{sec:hgap:sep-btype,sec:hgap:forbid-hslice,sec:hgap:2d-to-dd}
describe how to convert a shelf-based $\delta$-fractional packing $P$
of $\Ihat$ having $t$ distinct shelf heights into a shelf-based $d$D bin packing of $I$.
We call this conversion algorithm $\inflate$.

It is easy to see that the time taken by $\inflate$ is $O(|I|d\log|I|)$.

If $P$ has $m$ bins, then the number of bins in $\inflate(P)$ is at most
\[ m + t(Q-1) + \ceil{\frac{\delta(Q+m)}{1-\delta}}
< \frac{m}{1-\delta} + t(Q-1) + 1 + \frac{\delta Q}{1-\delta}. \]
This proves \cref{thm:hgap:inflate}.

\subsection{Improving Running Time}
\label{sec:hgap:improve-time}

For simplicity of presentation, we left out some opportunities for improving the running
time of $\hgapk$. Here we briefly describe a way of speeding up $\hgapk$
which reduces its running time from $O(N^{1+\ceildeltsq}n^{R+2\ceildeltsq} + Nd + nd\log n)$ to
$O(N^{1+\ceildeltsq}n^{2\ceildeltsq} + Nd + nd\log n)$.
Here $N \defeq |\flatten(\Icalhat)|$, $n \defeq |\Icalhat|$, $\delta \defeq \eps/(2+\eps)$
and $R \defeq \binom{\ceildeltsq + \ceil{1/\delta}-1}{\ceil{1/\delta}-1}
\le (1+\ceildeltsq)^{1/\delta}$.

In $\guessShelves$, we guess two things simultaneously:
(i) the number and heights of shelves
(ii) the packing of the shelves into bins.
This allows us to guess the optimal structured $\delta$-fractional packing.
But we don't need that; an approximate structured packing would do.

Therefore, we only guess the number and heights of shelves.
We guess at most $N^{\ceildeltsq}+1$ distinct heights of shelves,
and by \cref{thm:nos-sum}, we guess at most $(n+1)^{\ceildeltsq}$ vectors
of shelf-height frequencies.
Then we can use Lueker and Fernandez de la Vega's $O(n\log n)$-time APTAS for
1BP~\cite{bp-aptas} to pack the shelves into bins.

Also, once we guess the distinct heights of shelves,
we don't need to run $\chooseAndPack$ afresh for every packing of empty shelves.
We can reuse the dynamic programming table.

The running time is, therefore,
\begin{align*}
& O\left(N^{\ceildeltsq}\left(n^{\ceildeltsq} n\log n + Nn^{2\ceildeltsq}\right) + Nd + nd\log n\right)
\\ &= O(N^{1+\ceildeltsq}n^{2\ceildeltsq} + Nd + nd\log n).
\end{align*}

\section{\texorpdfstring{$\hdhkunit$}{HDH-unit-pack}}
\label{sec:hdhkunit}

This section gives a precise description of $\hdhkunit$ (see \cref{sec:hdhk-prelims:hdhkunit})
and proves its correctness.

\subsection{Shelf-Based Packing}
\label{sec:shelf}

A packing of 2D items in a bin (or strip) is said to be \emph{shelf-based} iff
the bin can be decomposed into regions, called shelves, using horizontal cuts,
and the bottom edge of each item touches the bottom edge of some shelf.
See \cref{fig:shelf-based} for an example.
Next-Fit Decreasing Height (NFDH) and First-Fit Decreasing Height (FFDH)%
~\cite{coffman1980performance} are well-known shelf-based algorithms for 2BP and 2SP.

\begin{figure}[!ht]
\centering
\begin{tikzpicture}[scale=0.8,
shelf-line/.style = {draw={black!30}},
item/.style = {fill={black!10}}
]
\draw[shelf-line] (0,1.5) -- (4,1.5);
\draw[shelf-line] (0,2.5) -- (4,2.5);
\draw (0,0) rectangle (4,4);
\draw[item]
    (0,0) rectangle +(1.25,1.5)
    (1.25,0) rectangle +(1.5,1)
    (2.75,0) rectangle +(1,0.5);
\draw[item]
    (0,1.5) rectangle +(1.5,1)
    (1.5,1.5) rectangle +(1.75,0.75);
\draw[item]
    (0,2.5) rectangle +(2,0.6)
    (2,2.5) rectangle +(1,0.5);
\end{tikzpicture}

\caption{An example of shelf-based packing for $d=2$ with 3 shelves.}
\label{fig:shelf-based}
\end{figure}
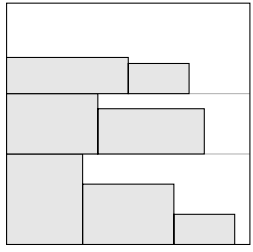

The definition of shelf-based packing can be extended to $d$D for $d \ge 1$.
For $d=1$, every packing is said to be a shelf-based packing.
For $d \ge 2$, for a $d$D cuboid, there are two faces of the cuboid that are
perpendicular to the $d\Th$ dimension. The face with the smaller $d\Th$ coordinate
is called the \emph{base} of the cuboid.
A packing of $d$D items into a bin is shelf-based iff the $d$D bin
can be split into $d$D shelves using hyperplanes perpendicular to the $d\Th$ dimension,
and the base of each item is placed on the base of some shelf.

A packing of $d$D items into a bin is \emph{recursive-shelf-based} iff
the packing is shelf-based and the packing of the bases of items on the base of each shelf
is a $(d-1)$D recursive-shelf-based bin packing.
(For $d=1$, every packing is said to be recursive-shelf-based.)

This helps us reduce $d$BP to $(d-1)$BP, $(d-1)$BP to $(d-2)$BP, and so on.
The algorithm $\hdhk$ by Caprara~\cite{caprara2008}
outputs a recursive-shelf-based packing by using this strategy.

\subsection{Description and Analysis of \texorpdfstring{$\hdhkunit$}{HDH-unit-pack}}

For a $d$D item $i$, define $i^{(j)}$ as the $j$D item
obtained by ignoring all dimensions of $i$ other than the first $j$.
For a set $I$ of $d$D items, let $I^{(j)} \defeq \{i^{(j)}: i \in I\}$.

$\hdhkunit$ takes as input a set $I$ of $d$D items,
where all items in $I$ have the same type vector and $\vol(f_k(I - \{\last(I)\})) < 1$.
$\hdhkunit(I)$ works recursively on $d$.
When $d = 1$, it simply returns $I$.
When $d > 1$, it first sorts $I$ in decreasing order of $d\Th$ dimension
if $\type_k(\ell_d(i)) = k$ for each item $i \in I$.
It then repeatedly picks the smallest prefix $J$ of $I$ such that $\vol(f_k(J^{(d-1)})) \ge 1$,
and packs $J$ into a $d$D shelf.
It packs all those shelves into a $d$D bin and returns that packing.
See \cref{algo:hdhkunit} for a more precise description.

\begin{algorithm}[!ht]
\caption{$\hdhkunit^{[t]}(I)$:
For any $d \ge 1$, returns a recursive-shelf-based packing of $I$ into a $d$D bin,
where $I$ is a sequence of $d$D cuboidal items
and $\vol(f_k(I-\{\last(I)\})) < 1$.
Here $\last(I)$ is the last item in sequence $I$.
Also, all items in $I$ have the same type $t$, i.e.,
$\forall i \in I, \type(i) = t$.}
\label{algo:hdhkunit}
\begin{algorithmic}[1]
\If{$d \texttt{ == } 1$}
    \Comment{when items are 1D}
    \State \Return $I$.
    \Comment{\cref{thm:hdhkunit} proves that they fit in a bin.}
\EndIf
\If{$t_d \texttt{ == } k$}
    \Comment{when length in $d\Th$ dimension is small}
    \State \label{alg-line:hdhunit:sort}Sort $I$ in decreasing order of $d\Th$ dimension.
\EndIf
\Comment{otherwise don't disturb ordering of items.}
\State Let $P$ be an empty list.
\While{$|I| > 0$}
    \State \label{alg-line:hdhunit:prefix}Find $J$, the smallest prefix of $I$ such that
        $J = I$ or $\vol(f_k(J^{(d-1)})) \ge 1$.
    \State Let $t'$ be a $(d-1)$-dimensional vector obtained by
        removing the $d\Th$ entry from $t$.
    \State \label{alg-line:hdhunit:recurse}$S = \hdhkunit^{[t']}(J^{(d-1)})$
        \Comment{$S$ is a $d$D shelf containing items $J$.}
    \State Append $S$ to the list $P$.
    \State Remove $J$ from $I$.
\EndWhile
\State Return the shelf packing $P$.
\Comment{\cref{thm:hdhkunit} proves that the sum of heights of shelves doesn't exceed 1,
so this is a valid packing.}
\end{algorithmic}
\end{algorithm}

Define $\wfk(i)$ to be the cuboid $\itild$ where $\ell_j(\itild) \defeq f_k(\ell_j(i))$
for $j \in [d-1]$ and $\ell_d(\itild) \defeq \ell_d(i)$.
Define $\wfk(I) \defeq \{\wfk(i): i \in I\}$.

\begin{theorem}[Correctness]
\label{thm:hdhkunit}
For a set $I$ of $d$D cuboidal items, if $\vol(f_k(I-\{\last(I)\})) < 1$,
then $\hdhkunit(I)$ can pack $I$ in a $d$D bin.
\end{theorem}
\begin{proof}
Let us prove this by induction on $d$.
Let $\mathcal{P}(d)$ be this proposition:
For every sequence $I$ of $d$D items, if $\vol(f_k(I)) < 1 + \vol(f_k(\last(I)))$,
then $\hdhkunit(I)$ can pack $I$ into a $d$D bin.

\textbf{Base case}:
Let $I$ be a sequence of 1D items such that $\vol(f_k(I)) < 1 + \vol(f_k(\last(I)))$.

Suppose $t_1 \neq k$.
Then for all $i \in I, \vol(i) \le 1/t_1 = \vol(f_k(i))$. Therefore,
\begin{align*}
& \vol(f_k(I)) < 1 + \vol(f_k(\last(I)))
\\ &\implies \frac{|I|}{t_1} < 1 + \frac{1}{t_1}
\\ &\implies |I| < t_1 + 1 \implies |I| \le t_1
\\ &\implies \vol(I) \le \frac{|I|}{t_1} \le 1.
\end{align*}
Since $\vol(I) \le 1$, $I$ fits in a bin.

Suppose $t_1 = k$.
Then for all $i \in I, \vol(f_k(i)) = \frac{k}{k-2}\vol(i)$. Therefore,
\begin{align*}
\vol(I) &= \frac{k-2}{k}\vol(f_k(I))
\\ &< \frac{k-2}{k}\left( 1 + \vol(f_k(\last(I)))\right)
\\ &= \frac{k-2}{k} + \vol(\last(I))
\\ &\le \frac{k-2}{k} + \frac{1}{k} < 1.
\end{align*}
Since $\vol(I) \le 1$, $I$ fits in a bin.
Therefore, $\mathcal{P}(1)$ holds.

\textbf{Inductive step}:\\
Let $d \ge 2$ and assume $\mathcal{P}(d-1)$ holds.
Let $I$ be a sequence of $d$D items such that $\vol(f_k(I)) < 1 + \vol(f_k(\last(I)))$.
$\mathcal{P}(d-1)$ implies that $\hdhkunit(I)$ doesn't fail at
line \ref{alg-line:hdhunit:recurse}.
Let $s \defeq \last(I)$.

For $i \in I$, define $w(i) \defeq \prod_{j=1}^{d-1} f_k(\ell_j(i))$
and for $X \subseteq I$, define $w(X) \defeq \sum_{i \in X} w(i)$.
Let there be $p$ shelves in the list $P$.
Let $S_j$ be the $j\Th$ shelf that was added to $P$.
Given the way each prefix is chosen in line \ref{alg-line:hdhunit:prefix},
\begin{equation}\label{eqn:hdhunit:shelf-wide}
\forall j \le p-1, w(S_j) \ge 1 \end{equation}
Define $\ell_d(S_j) \defeq \max_{i \in S_j} \ell_d(i)$ to be the height of shelf $S_j$.
Let $H$ be the total height of the shelves, i.e. $H \defeq \sum_{j=1}^p \ell_d(S_j)$.
Then we need to prove that the shelves fit in the bin, i.e. $H \le 1$.

\textbf{Case 1}: Suppose $t_d \neq k$.\\
Then $\forall i \in I, \ell_d(i) \le 1/t_d = f_k(\ell_d(i))$. Therefore,
\[ 1 > \vol(f_k(I-s)) = \frac{w(I-s)}{t_d} \implies w(I-s) < t_d. \]
Since ordering of items is not disturbed, $s \in S_p$. Therefore,
\begin{align*}
& t_d > w(I-s) = \sum_{j=1}^{p-1} w(S_j) + w(S_p - s) \ge p-1
\tag{by (\ref{eqn:hdhunit:shelf-wide})}
\\ &\implies p < t_d + 1 \implies p \le t_d
\\ &\implies H = \sum_{j=1}^p \ell_d(S_j) \le \frac{p}{t_d} \le 1.
\tag{$\forall i \in I, \ell_d(i) \le 1/t_d$}
\end{align*}
Since $H \le 1$, the shelves fit in a $d$D bin.

\textbf{Case 2}: Suppose $t_d = k$.\\
Then $\forall i \in I, f_k(\ell_d(i)) = \frac{k}{k-2}\ell_d(i)$. Therefore,
\begin{align}
\vol(\wfk(I)) &= \sum_{i \in I} w(i)\ell_d(i)
= \frac{k-2}{k} \sum_{i \in I} w(i)f_k(\ell_d(i))
= \frac{k-2}{k} \vol(f_k(I))
\nonumber
\\ \label{eqn:hdhunit:wfk-last}
&< \frac{k-2}{k} (1 + \vol(f_k(s)))
= \frac{k-2}{k} + \vol(\wfk(s)).
\end{align}
Since items in $I$ were sorted in decreasing order of $\ell_d$
(line \ref{alg-line:hdhunit:sort}),
$\forall i \in S_j, \ell_d(i) \ge \ell_d(S_{j+1})$.
Then by (\ref{eqn:hdhunit:shelf-wide}), we get that for all $j \in [p-1]$,
\begin{equation}
\label{eqn:hdhunit:wfk-lb-height}
\vol(\wfk(S_j)) \ge w(S_j)\ell_d(S_{j+1}) \ge \ell_d(S_{j+1})
\end{equation}
Therefore,
\begin{align*}
H &= \sum_{j=1}^{p} \ell_d(S_j)
\le \frac{1}{k} + \sum_{j=1}^{p-1} \ell_d(S_{j+1})
\tag{since $\ell_d(S_1) \le 1/k$}
\\ &\le \frac{1}{k} + \sum_{j=1}^{p-1} \vol(\wfk(S_j))
\tag{by (\ref{eqn:hdhunit:wfk-lb-height})}
\\ &< \frac{1}{k} + \vol(\wfk(I))
\\ &< \frac{1}{k} + \frac{k-2}{k} + w(s)\ell_d(s)
\tag{by (\ref{eqn:hdhunit:wfk-last})}
\\ &\le \frac{k-1}{k} + \frac{1}{k} = 1.
\tag{since $\ell_d(s) \le 1/k$ and $w(s) = \vol(f_k(s^{(d-1)})) \le 1$}
\end{align*}
Since $H \le 1$, the shelves fit in a $d$D bin.
Therefore, $\mathcal{P}(d)$ holds.

Therefore, by mathematical induction, $\mathcal{P}(d)$ holds for all $d \ge 1$.

\end{proof}

Note that $\hdhkunit$ has a running time of $O(nd\log n)$.

\textbf{Comment on Caprara's~\cite{caprara2008} analysis of $\hdhk$.}
Caprara~\cite{caprara2008} implicitly proves \cref{thm:hdhkunit} in Lemma 4.1 in their paper
and their proof is less detailed than ours.
Their algorithm is similar to ours, except that
they allow arbitrarily reordering $I$ when $t_d \neq k$,
and instead of choosing a prefix of $I$
(line \ref{alg-line:hdhunit:prefix} in $\hdhkunitHyp$),
they choose a subset of $I$ that is minimal for some properties.

\section{Harmonic Algorithm for Strip Packing}
\label{sec:hdhk-sp}

\subsection{Multiple-Choice Strip Packing}

Let $I$ be a set of $d$D cuboidal items.
In the $d$D strip packing problem ($d$SP), we have to compute
a feasible packing of $I$ (without rotating the items)
into a $d$D cuboid (called a strip) that has length one in the first $d-1$ dimensions
and has the minimum possible length (called height) in the $d\Th$ dimension.
Let $\optdsp(I)$ denote the minimum height of a strip needed to pack $I$.

In the $d$D multiple-choice strip packing problem ($d$MCSP),
we are given as input a set $\Ical = \{I_1, I_2, \ldots, I_n\}$,
where for each $j$, $I_j$ is a set of items, called an {\em itemset}.
We have to pick exactly one item from each itemset and pack those items
into a strip of minimum height.

Equivalently, given an input instance $\Ical$, we have to select an
assortment $K \in \assortSet(\Ical)$ and output a strip packing of $K$,
such that the total height of the strip is minimized. Therefore,
$\optdmcsp(\Ical) \defeq \min_{K \in \assortSet(\Ical)} \optdsp(K)$.

\subsection{Revisiting the \texorpdfstring{$\hdhk$}{HDHk} Algorithm}

Caprara~\cite{caprara2008} gave an algorithm for $d$SP,
which we call $\hdhksp$.
We will first prove a few useful properties of $\hdhksp$
and then see how to extend it to $d$MCSP.

For a $d$D item $i$, $\btype(i)$ (called \emph{base type}) is defined to be
a $(d-1)$-dimensional vector whose $j\Th$ component is $\type_k(\ell_j(i))$.
Define $\wfk(i)$ to be the cuboid $\itild$ where $\ell_j(\itild) \defeq f_k(\ell_j(i))$
for $j \in [d-1]$ and $\ell_d(\itild) \defeq \ell_d(i)$.
Define $\wfk(I) \defeq \{\wfk(i): i \in I\}$.
Similarly define $\wHk(i)$ and $\wHk(I)$.
Define $i^{(j)}$ to be the $j$-dimensional item
obtained by ignoring all dimensions of $i$ other than the first $j$.
For a set $I$ of $d$D items, let $I^{(j)} \defeq \{i^{(j)}: i \in I\}$.

$\hdhksp$ works by first partitioning the items based on $\btype$.
Then for each partition, it repeatedly picks the smallest prefix $J$
such that $\vol(f_k(J^{(d-1)})) \ge 1$ and packs $J$ into a $d$D shelf
by using $\hdhkunit$ on $J^{(d-1)}$
(see \cref{sec:shelf} for the definition of a $d$D shelf).
See \cref{algo:hdhksp} for a more precise description of $\hdhksp$.
Note that $\hdhksp(I)$ has a running time of $O(nd\log n)$, where $n \defeq |I|$.

\begin{algorithm}[!ht]
\caption{$\hdhksp(I)$:
Returns a strip packing of $d$D items $I$ ($d \ge 2$).}
\label{algo:hdhksp}
\begin{algorithmic}[1]
\State Let $P$ be an empty list.
\For{each $\btype$ $t$}
    \State $I^{[t]} = \{i \in I: \btype(i) = t\}$.
    \State Sort items in $I^{[t]}$ in non-increasing order of their length in the $d\Th$ dimension.
    \While{$|I^{[t]}| > 0$}
        \State Find $J$, the smallest prefix of $I^{[t]}$ such that
            $J = I^{[t]}$ or $\vol(f_k(J^{(d-1)})) \ge 1$.
        \State $S = \hdhkunitHyp^{[t]}(J^{(d-1)})$
            \Comment{$S$ is a $d$D shelf containing items $J$.}
        \State Append $S$ to the list $P$.
        \State Remove $J$ from $I^{[t]}$.
    \EndWhile
\EndFor
\State Return the strip packing formed by the shelves $P$.
\end{algorithmic}
\end{algorithm}

\begin{theorem}
\label{thm:hdhk-fvol}
The height of the strip packing produced by $\hdhksp(I)$ is less than $Q + \vol(\wfk(I))$,
where $Q$ is the number of distinct $\btype$s of items (so $Q \le k^{d-1}$).
\end{theorem}
\begin{proof}
Let there be $p^{[q]}$ shelves of $\btype$ $q$ produced by $\hdhksp(I)$.
Let $S_j^{[q]}$ be the set of items in the $j\Th$ shelf of $\btype$ $q$.
Define $\ell_d(S_j^{[q]}) \defeq \max_{i \in S_j^{[q]}} \ell_d(i)$
to be the height of shelf $S_j^{[q]}$.

Since items in $I^{[q]}$ were sorted in decreasing order of $\ell_d$,
$\forall i \in S_j^{[q]}$, $\ell_d(i) \ge \ell_d(S_{j+1}^{[q]})$.
Given the way we choose prefixes, $\vol(f_k(S_j^{[q] (d-1)})) \ge 1$ for $j \in [p-1]$.
\begin{equation}
\label{eqn:wfk-lb-height}
\vol(\wfk(S_j^{[q]}))
\ge \vol(f_k(S_j^{[q] (d-1)}))\ell_d(S_{j+1}^{[q]})
\ge \ell_d(S_{j+1}^{[q]})
\end{equation}

Total height of the strip packing is
\begin{align*}
\sum_{q=1}^Q \sum_{j=1}^{p^{[q]}} \ell_d(S_j^{[q]})
&\le \sum_{q=1}^Q \left(1 + \sum_{j=1}^{p^{[q]}-1} \ell_d(S_{j+1}^{[q]})\right)
\tag{since $\ell_d(S_1^{[q]}) \le 1$}
\\ &\le Q + \sum_{q=1}^Q \sum_{j=1}^{p^{[q]}-1} \vol(\wfk(S_j^{[q]}))
\tag{by \eqref{eqn:wfk-lb-height}}
\\ &< Q + \sum_{q=1}^Q \sum_{j=1}^{p^{[q]}} \vol(\wfk(S_j^{[q]}))
= Q + \vol(\wfk(I)).
\qedhere \end{align*}
\end{proof}

\begin{theorem}
\label{thm:wfvol-sp}
For a set $I$ of $d$D items, $\vol(\wfk(I)) \le T_k^{d-1} \optdsp(I)$.
\end{theorem}
\begin{proof}
$I$ fits in a unit strip of height $\optdsp(I)$.
Let $I'$ be the items obtained by scaling each item's height by $1/\optdsp(I)$.
Then $I'$ fits in a unit cube.

Since $H_k$ is a \dff{}, $\wHk(I')$ fits in a unit cube by \cref{thm:dff-pack}.
Therefore, $\wHk(I)$ can be packed into a unit strip of height $\optdsp(I)$.
Therefore, $\vol(\wfk(I)) \le T_k^{d-1}\vol(\wHk(I)) \le T_k^{d-1}\optdsp(I)$.
\end{proof}

\begin{corollary}
\label{thm:hdhk-sp-appx}
$\hdhksp(I)$ packs $I$ into a strip of height less than $Q + T_k^{d-1}\optdsp(I)$,
where $Q$ is the number of distinct $\btype$s of items.
\end{corollary}
\begin{proof} Follows from \cref{thm:hdhk-fvol,thm:wfvol-sp}. \end{proof}

\subsection{Extending \texorpdfstring{$\hdhksp$ to $d$MCSP}{HDH-SP to dMCSP}}
\label{sec:hdhk-mcsp}

\begin{theorem}
Let $\Ical$ be a $d$MCSP instance.
Let $\Khat \defeq \{\argmin_{i \in I} \vol(\wfk(i)): I \in \Ical\}$,
i.e., $\Khat$ is the assortment obtained by picking from each itemset
the item $i$ having the minimum value of $\vol(\wfk(i))$.
Then the height of the strip packing produced by $\hdhksp(\Khat)$
is less than $Q + T_k^{d-1}\optdmcsp(\Ical)$,
where $Q$ is the number of distinct $\btype$s of items in $\flatten(\Ical)$
(so $Q \le k^{d-1}$).
\end{theorem}
\begin{proof}
For any assortment $K$, $\vol(\wfk(\Khat)) \le \vol(\wfk(K))$.
Let $K^*$ be the assortment in an optimal packing of $\Ical$.
By \cref{thm:hdhk-fvol,thm:wfvol-sp}, the height of the strip packing produced by
$\hdhksp(\Khat)$ is less than
\[ Q + \vol(\wfk(\Khat))
\le Q + \vol(\wfk(K^*))
\le Q + T_k^{d-1}\optdsp(K^*)
= Q + T_k^{d-1}\optdmcsp(\Ical).
\qedhere \]
\end{proof}

Let $N \defeq |\flatten(\Ical)|$ and $n \defeq |\Ical|$.
Then we can find $\Khat$ in $O(Nd)$ time and compute $\hdhksp(\Khat)$ in $O(nd\log n)$ time.
Therefore, we get a $T_k^{d-1}$-\asymAppx{} algorithm for $d$MCSP
that runs in $O(Nd + nd\log n)$ time.

\section{Harmonic Algorithm for \safed{}MCKS}
\label{sec:hdhks}

In the $d$D knapsack problem ($d$KS), we are given a set $I$ of $d$D items,
and a profit $p(i)$ for each item $i \in I$. We have to compute
a maximum-profit packing of a subset of $I$ (without rotating the items)
into a $d$D unit cube (called a knapsack).

In the $d$D multiple-choice knapsack problem ($d$MCKS),
we are given a set $\Ical = \{I_1, I_2, \ldots, I_n\}$ as input,
where for each $j$, $I_j$ is a set of items, called an {\em itemset},
and each item $i \in I_j$ has a profit $p(i)$.
We have to pick at most one item from each itemset and pack those items
into a $d$D bin such that the total profit is maximized.

For a $d$D item $i$, $\btype(i)$ (called \emph{base type}) is defined to be
a $(d-1)$-dimensional vector whose $j\Th$ component is $\type_k(\ell_j(i))$.
Define $\wfk(i)$ to be the cuboid $\itild$ where $\ell_j(\itild) \defeq f_k(\ell_j(i))$
for $j \in [d-1]$ and $\ell_d(\itild) \defeq \ell_d(i)$.
Define $\wfk(I) \defeq \{\wfk(i): i \in I\}$.
Similarly define $\wHk(i)$ and $\wHk(I)$.

We will see a fast and simple algorithm $\hdhknf(I)$ (\cref{algo:hdhknf}) for $d$BP
that we will use to design an algorithm for $d$MCKS.

\begin{algorithm}[!ht]
\caption{$\hdhknf(I)$:
Returns a bin packing of $d$D items $I$ ($d \ge 2$).}
\label{algo:hdhknf}
\begin{algorithmic}[1]
\State Let $P$ be the list of shelves output by $\hdhkspHyp(I)$.
\Comment{cf.~\cref{sec:hdhk-sp} for $\hdhksp$.}
\State Let $P'$ be an empty list.
\For{each $\btype$ $q$}
    \State Let $S_1^{[q]}, S_2^{[q]}, \ldots, S_{p^{[q]}}^{[q]}$
        be the shelves in $P$ of $\btype$ $q$, in decreasing order of height.
    \State Pack $S_1^{[q]}$ in a $d$D bin.
    \State For $j \ge 2$, add $S_j^{[q]}$ to $P'$.
\EndFor
\State Interpreting each shelf $S_j^{[q]}$ in $P'$ as a 1D item of size $\ell_d(S_j^{[q]})$,
    pack the shelves into $d$D bins using Next-Fit.
\end{algorithmic}
\end{algorithm}

Note that $\hdhknf(I)$ runs in $O(nd\log n)$ time.

\begin{theorem}
\label{thm:hdhk-nf-fvol}
$\hdhknf(I)$ uses at most $Q + \smallceil{2\vol(\wfk(I))}$ bins,
where $Q$ is the number of distinct $\btype$s of items.
\end{theorem}
\begin{proof}
For each $q \in [Q]$, $S_1^{[q]}$ occupies one bin.

As per \cref{eqn:wfk-lb-height} in the proof of \cref{thm:hdhk-fvol},
for all $t \le p^{[q]}-1$, we get $\vol(\wfk(S_t^{[q]}))
\ge \ell_d(S_{t+1}^{[q]})$.

Let $H$ be the total height of the shelves in $P'$. Then
\begin{align*}
H &= \sum_{q=1}^Q \sum_{t=1}^{p^{[q]}-1} \ell_d(S_{t+1}^{[q]})
\le \sum_{q=1}^Q \sum_{t=1}^{p^{[q]}-1} \vol(\wfk(S_t^{[q]}))  \tag{by \eqref{eqn:wfk-lb-height}}
\\ &< \sum_{q=1}^Q \sum_{t=1}^{p^{[q]}} \vol(\wfk(S_t^{[q]}))
= \vol(\wfk(I)).
\end{align*}
Next-Fit guarantees that for a 1BP instance $J$, number of bins used is at most $\ceil{2\vol(J)}$.
So for the shelves in $P'$, we use $\ceil{2H}$ bins.
The total number of bins used is therefore
$Q + \ceil{2H} \le Q + \smallceil{2\vol(\wfk(I))}$.
\end{proof}

By \cref{thm:wfvol-sp,thm:hdhk-nf-fvol}, we get that $\hdhknf$ is $2T_k^{d-1}$-\asymAppx{}.

Lawler gave an FPTAS for 1MCKS that has a running time of
$O(N\log N + Nn/\eps)$~\cite{lawler1979fast},
where $N \defeq |\flatten(\Ical)|$ and $n \defeq |\Ical|$.
We will use it along with $\hdhknf[3]$ to get an algorithm for $d$MCKS,
called $\hdhks$ (see \cref{algo:hdhks}).

Our algorithm for $d$MCKS, called $\hdhks(\Ical)$, works as follows:
It computes a 1MCKS instance $\Icalhat$ by replacing each item $i$ in $\Ical$
by a 1D item $\vol(\wHk[3](i))$.
It uses the FPTAS for 1MCKS to obtain a $(1-\eps)$-\appx{} solution $J$ to $\Icalhat$.
It uses $\hdhknf[3]$ to pack the corresponding $d$D items of $J$ into bins.
It then selects the most profitable bin.
See \cref{algo:hdhks} for a more detailed description.

\begin{algorithm}[!ht]
\caption{$\hdhks(\Ical)$: algorithm for $d$MCKS.}
\label{algo:hdhks}
\begin{algorithmic}[1]
\State $\Icalhat = \{\{\vol(\wHk[3](i)): i \in I\}: I \in \Ical\}$.
\Comment{Reduction to 1MCKS.}
\State Let $\Jhat$ be a $(1-\eps)$-\appx{} solution to the 1MCKS instance $\Icalhat$
    output by the FPTAS for 1MCKS.
\State Let $J$ be the items of $\Ical$ corresponding to $\Jhat$.
\State Let $[J_1, J_2, \ldots, J_b]$ be the bin packing of $J$ produced using $\hdhknf[3]$.
\State ${\displaystyle j_{\max} = \argmax_{j=1}^b p(J_j)}$
\State \Return $J_{j_{\max}}$.
\end{algorithmic}
\end{algorithm}

$\hdhks$ runs in $O(Nd + N\log N + Nn/\eps + nd\log n)$ time.

\begin{theorem}
\label{thm:emb-ks}
$\hdhks$ is $(1-\eps)3^{-d}$-\appx{}.
\end{theorem}
\begin{proof}
Let $I$ be a set of $d$D items.
Suppose $S \subseteq I$ can be packed into a bin.
Then by \cref{thm:dff-pack}, $\widehat{S} = \{\vol(\wHk[3](i)): i \in S\}$
can also be packed into a bin. Therefore, $\optdmcks[1](\Icalhat) \ge \optdmcks(\Ical)$.

The FPTAS for 1MCKS gives us $\Jhat$ such that
$p(\Jhat) \ge (1-\eps)\optdmcks[1](\Icalhat)$.
$\hdhknf$ packs $J$ into
$b \le 3^{d-1} + \smallceil{2T_3^{d-1}\vol(\wHk[3](J))} \le 3^d$ bins.
Given the way we choose $j_{\max}$,
\[ p(J_{j_{\max}}) \ge \frac{p(J)}{b} = \frac{p(\Jhat)}{b}
\ge \frac{(1-\eps)\optdmcks[1](\Icalhat)}{b} \ge \frac{1-\eps}{3^d}\optdmcks(\Ical).
\qedhere \]
\end{proof}

\section{\DFF{} Transform}
\label{sec:dff-trn}

In this section, we prove \cref{thm:dff-pack}.

\begin{lemma}
\label{thm:dff-pack-1}
Let $I$ be a set of $d$D items that can be packed into a bin. Let $g$ be a \dff{}.
Let $q \in [d]$. For $i \in I$, define $g(i)$ to be the item $\ihat$ for which
$\ell_j(\ihat) \defeq \ell_j(i)$ when $j \neq q$ and $\ell_q(\ihat) \defeq g(\ell_q(i))$.
Then the items $\{g(i): i \in I\}$ can be packed into a $d$D bin (without rotating the items).
\end{lemma}
Bansal, Caprara and Sviridenko~\cite{rna} give a brief proof sketch for $d=2$,
based on which we provide a full proof below.
\begin{proof}
Any $d$D cuboid can be represented as the Cartesian product of
$d$ closed intervals on the real line. Let the bin be $[0, 1]^d$.
Any item $i \in I$ can be written as $\prod_{j=1}^d [v_j(i), v_j(i) + \ell_j(i)]$.
Here $v_j(i)$ is called the \emph{position} of item $i$ in dimension $j$.
Since each item $i$ lies completely inside the bin, $0 \le v_j(i) < v_j(i) + \ell_j(i) \le 1$.
Two cuboids $A$ and $B$ are said to overlap if their intersection has positive volume.
Since $I$ is a valid packing, no two items overlap.

Assume \wLoG{} that $q = d$.
Let $\proj(i)$ be the projection of item $i$ onto the hyperplane
perpendicular to the $d\Th$ dimension.
This hyperplane can be thought of as the \emph{base} of the bin.

We will now show that for each item $i$, we can change $\ell_d(i)$ to $g(\ell_d(i))$
and change $v_d(i)$ so that the items continue to fit in the bin.
But to define what the new value of $v_d(i)$ would be,
we need to first introduce some notation.

For two items $i_1$ and $i_2$, we say that $i_1 \prec i_2$
($i_1$ is a predecessor of $i_2$)
iff $v_d(i_1) < v_d(i_2)$ and $\proj(i_1)$ overlaps $\proj(i_2)$.
Call a sequence $[i_0, i_1, \ldots, i_{m-1}]$ of items a \emph{chain}
iff $i_{m-1} \prec i_{m-2} \prec \ldots \prec i_0$. $i_0$ is called the head of this chain.
The \emph{augmented height} of a chain $S$ is defined to be
$\sum_{i \in S} g(\ell_d(i))$.
For each item $i$, we wish to find the chain headed at $i$
with the maximum augmented height.

For an item $i$, define
\[ \level(i) \defeq \begin{cases} 0 & \textrm{if } i \textrm{ has no predecessors}
\\ {\displaystyle 1 + \max_{i' \prec i} \level(i')} & \textrm{otherwise} \end{cases}. \]
Since $\prec$ is anti-symmetric, $\level$ is well-defined.
Define $\pi$ and $u$ as
\begin{align*}
u(i) &\defeq g(\ell_d(i)) + \begin{cases} 0 & \textrm{if } \level(i) = 0
    \\ u(\pi(i)) & \textrm{otherwise}\end{cases}
& \pi(i) &\defeq \begin{cases} \Null & \textrm{if } \level(i) = 0
\\ {\displaystyle \argmax_{i' \prec i} u(i')} & \textrm{otherwise} \end{cases}.
\end{align*}
In the definition of $\pi$, ties can be broken arbitrarily for $\argmax$.
$i' \prec i$ implies $\level(i') < \level(i)$, so $\level(\pi(i)) < \level(i)$.
This ensures that the definitions of $\pi$ and $u$ are not mutually circular.

It can be proven, by inducting on $\level(i)$, that
$\Pi(i) \defeq [i, \pi(i), \pi(\pi(i)), \ldots]$ is the chain headed at $i$ with
the maximum augmented height, and that the augmented height of $\Pi(i)$ is $u(i)$.

\begin{transformation}
\label{trn:ld-and-vd}
For each item $i \in I$, change $\ell_d(i)$ to $g(\ell_d(i))$
and change $v_d(i)$ to $v_d'(i) \defeq u(i) - g(\ell_d(i))$.
\end{transformation}

We need to prove that \cref{trn:ld-and-vd} produces a valid packing,
i.e. items don't overlap and all items lie completely inside the bin $[0, 1]^d$.

Let $i_1$ and $i_2$ be any two items. We will prove that they don't overlap in the new packing.
If $\proj(i_1)$ and $\proj(i_2)$ don't overlap, then $i_1$ and $i_2$ don't overlap
and we are done, so assume $\proj(i_1)$ and $\proj(i_2)$ overlap.
Assume \wLoG{} that $i_1 \prec i_2$. Then $\level(i_2) \ge 1$ and
\[ v_d'(i_2) = u(i_2) - g(\ell_d(i_2)) = \max_{i' \prec i_2} u(i')
\ge u(i_1) = v_d'(i_1) + g(\ell_d(i_1)). \]
Therefore, $i_1$ and $i_2$ don't overlap in the new packing.

After \cref{trn:ld-and-vd}, item $i$ lies completely inside the bin
iff $v_d'(i) + g(\ell_d(i)) = u(i) \le 1$. Let $i_0 \defeq i$ and
$\Pi(i) = [i_0, i_1, i_2, \ldots, i_{m-1}]$.
Then $u(i) = \sum_{j=0}^{m-1} g(\ell_d(i_j))$ and for all $j \in [m-1], i_j \prec i_{j-1}$.
Since $i_j$ and $i_{j-1}$ don't overlap in the original packing,
but $\proj(i_j)$ and $\proj(i_{j-1})$ overlap,
we get $v_d(i_{j-1}) \ge v_d(i_j) + \ell_d(i_j)$. Therefore,
\begin{align*}
\sum_{j=0}^{m-1} \ell_d(i_j)
&\le \ell_d(i) + \sum_{j=1}^{m-1} (v_d(i_{j-1}) - v_d(i_j))
\tag{since $v_d(i_{j-1}) \ge v_d(i_j) + \ell_d(i_j)$}
\\ &= \ell_d(i) + v_d(i) - v_d(i_{m-1})
\le 1.
\tag{$\because$ in the original packing, $i$ lies in the bin}
\end{align*}
Since $g$ is a \dff{} and $\sum_{j=0}^{m-1} \ell_d(i_j) \le 1$,
we get $u(i) = \sum_{j=0}^{m-1} g(\ell_d(i_j)) \le 1$.
Therefore, the packing obtained by \cref{trn:ld-and-vd} is valid.
So $\{g(i): i \in I\}$ can be packed into a bin.
\end{proof}

\rthmDffPack*
\begin{proof}[Proof]
Apply \cref{thm:dff-pack-1} multiple times, with $q$ ranging from $1$ to $d$.
\end{proof}

\end{document}